\documentclass[12pt]{article}

\normalsize
\usepackage{epsfig}
\usepackage{amsfonts}
\usepackage{amssymb}
\usepackage{siunitx}
\usepackage{subfigure}
\usepackage{comment}
\usepackage{multirow}
\usepackage[cmex10]{amsmath}
\usepackage{bm}
\usepackage{setspace}
\usepackage{fullpage}
\usepackage{enumerate}
\usepackage{psfrag}
\usepackage{booktabs}
\usepackage{xspace}

\usepackage{amsthm}
\newtheorem{proposition}{Proposition}

\newtheorem{observation}{Observation}

\usepackage{amsmath}
\usepackage{mathtools}

\usepackage[linesnumbered,vlined]{algorithm2e}

\makeatletter
\def\BState{\State\hskip-\ALG@thistlm}
\makeatother

\usepackage{caption3} 
\DeclareCaptionOption{parskip}[]{} 
\usepackage[scriptsize]{caption}

\newcommand{\field}[1]{\mathcal{#1}}

\newcommand{\red}[1]{{\color{red} #1}}

\newcommand{\blue}[1]{{\color{blue} #1}}

\newcommand{\WLAN}[1]{WLAN$_\text{#1}$}
\def \WLANi{WLAN$_i$\xspace}

\newcommand{\Cs}{\mathcal{C}}
\newcommand{\Ss}{\mathcal{S}}

\newcommand{\Qmat}{\mathbf{Q}}

\newcommand{\piv}{\boldsymbol{\pi}}
\newcommand{\At}{\text{A}}
\newcommand{\Bt}{\text{B}}
\newcommand{\Xt}{\text{X}}

\begin{document}

\title{Analysis of Dynamic Channel Bonding \\ in Dense Networks of WLANs}



\date{}

\author{Azadeh Faridi, Boris Bellalta, Alessandro Checco\thanks{A.~Faridi and B.~Bellalta are with Universitat Pompeu Fabra, Barcelona; A.~Checco is with Trinity College Dublin, Ireland. Corresponding author: A. Faridi. e-mail: azadeh.faridi@upf.edu}}

\maketitle

\begin{abstract}
Dynamic Channel Bonding (DCB) allows for the dynamic selection and use of multiple contiguous basic channels in Wireless Local Area Networks (WLANs). A WLAN operating under DCB can enjoy a larger bandwidth, when available, and therefore achieve a higher throughput. However, the use of larger bandwidths also increases the contention with adjacent WLANs, which can result in longer delays in accessing the channel and consequently, a lower throughput. In this paper, a scenario consisting of multiple WLANs using DCB and operating within carrier-sensing range of one another is considered. An analytical framework for evaluating the performance of such networks is presented. The analysis is carried out using a Markov chain model that characterizes the interactions between  adjacent WLANs with overlapping channels. An algorithm is proposed for systematically constructing the Markov chain corresponding to any given scenario. The analytical model is then used to highlight and explain the key properties that differentiate DCB networks of WLANs from those operating on a single shared channel. Furthermore, the  analysis is applied to networks of IEEE 802.11ac WLANs operating under DCB--which do not fully comply with some of the simplifying assumptions in our analysis--to show that the analytical model can give accurate results in more realistic scenarios.


\end{abstract}





\section{Introduction}

Wireless Local Area Networks (WLANs) operate in the Industrial, Scientific, and Medical (ISM) radio bands. Licensed operation in ISM bands is reserved for ISM devices. Therefore, WLANs operating in these bands do so in the unlicensed regime and, in most cases, autonomously, as they may belong to different entities. Therefore, centralized spectrum allocation and interference management is generally not a feasible option when several WLANs operate within carrier-sensing range of one another. The CSMA/CA (Carrier-Sense Multiple Access with Collision Avoidance) mechanism used in currently deployed WLANs allows for decentralized sharing of the unlicensed bandwidth between different WLANs. This CSMA/CA mechanism uses a static channel allocation for all transmissions.\footnote{This is true for all of the amendments prior to IEEE 802.11ac.} Namely, the access point of a WLAN selects a channel during the initialization phase, and all nodes in the WLAN use the selected channel for all their subsequent transmissions, regardless of whether or not other channels are available. This leads to an under-utilization of the available bandwidth resources, which could otherwise be exploited to improve the performance of WLANs.

Channel bonding is a technique whereby the participating nodes are allowed to use a contiguous set of available channels for their transmissions, thus potentially achieving a higher throughput \cite{deek2014intelligent}. However, using wider channels also increases the contention between the neighboring nodes, unless the contending nodes are allowed to choose their channels dynamically, based on the instantaneous spectrum occupancy at each transmission. This technique is usually referred to as Dynamic Channel Bonding (DCB) \cite{Gong2011ChannelBounding2,park2011ieee,bellaltachannel}.

If a WLAN using DCB is operating in isolation, all its transmitting nodes will have access to the same channel. In such a scenario, for every transmission of a given target node, either the entire channel is available, in which case the target node will use the entire channel for transmission, or the entire channel is in use by another node, in which case the target node will have to defer its transmission. Therefore, a single WLAN operating under DCB is exactly equivalent to a classical CSMA/CA network using a wider channel. However, when several WLANs operate in the vicinity of each other in a decentralized manner, they may be operating on different, but possibly overlapping, channels. Due to the possibility of partial channel overlap between different WLANs in this case, the selected set of channels by a node in a given target WLAN will be the largest contiguous set not in use by nodes in other WLANs. 


The goal of this paper is to formally analyze the interactions between neighboring WLANs operating under DCB, to which we will hereafter refer as a \textit{DCB network of WLANs} or simply a \textit{DCB network}. We consider a generalized form of DCB, and therefore, our analysis can be easily tuned to particular implementations of DCB, including the one proposed in the IEEE 802.11ac amendment \cite{IEEE80211ac} and other future amendments, such as IEEE 802.11ax \cite{bellalta2015WCM}. Our analysis also applies to Static Channel Bonding (SCB), another channel bonding technique proposed in the IEEE 802.11ac amendment, in which transmissions are only allowed when the entire assigned channel width is free.

We analyze a DCB network of WLANs using a continuous-time Markov chain (CTMC) model. With each WLAN having several choices for its transmission channel, the number of possible states in the constructed CTMC grows rapidly with the number of participating WLANs. Furthermore, since in DCB, nodes choose the largest available channel, some states and transitions that may seem feasible at first are not actually so. Therefore, constructing the CTMC for a given scenario is not a trivial task. We devise an algorithm that constructs the CTMC corresponding to any given scenario and channel configuration.

The CTMC modeling relies on having exponentially distributed backoff and transmission times. In single-channel CSMA/CA networks, it is shown that the state probabilities are insensitive to the backoff and transmission time distributions \cite{liew2010back,van2010insensitivity}. We show that in DCB networks, the insensitivity property does not hold. However, the impact of this sensitivity on expected throughput is quite low. Comparing the results obtained from our analytical model to those obtained through simulation, we show that our analysis can be used to calculate, with high accuracy, the throughput of IEEE 802.11 WLANs with realistic (non-exponential) backoff and transmission time distributions.

We furthermore observe that, in contrast to CTMCs corresponding to single-channel CSMA/CA networks \cite{liew2010back, laufercapacity} and SCB networks \cite{bellalta2015interactions}, the CTMCs corresponding to DCB networks are not necessarily reversible \cite{kelly2011reversibility}. Therefore, many properties of those systems, including the aforementioned insensitivity to backoff and transmission time distributions, do not hold. In general, non-reversible CTMCs are more difficult to deal with since local balance does not hold for them. Despite that, we use our analytcal model to provide intuitive explanations about the behavior of the CTMCs corresponding to DCB networks. 

The use of CTMC models for the analysis of CSMA/CA networks was originally developed in~\cite{boorstyn1987throughput} and was further extended in the context of IEEE 802.11 networks in~\cite{wang2005throughput,durvy2006packing,liew2010back,nardelli2012closed,laufercapacity}, among others. Although the CTMC modeling of the IEEE 802.11 backoff mechanism is less detailed than Bianchi's well-known DTMC model~\cite{bianchi2000performance}, it offers greater versatility in modeling a broad range of topologies. Moreover, experimental results \cite{liew2010back,nardelli2012closed} demonstrate that CTMC models provide remarkably accurate throughput estimates for actual IEEE 802.11 systems. CTMCs have  also been used to model the interactions between neighboring WLANs operating under SCB in \cite{bellalta2015interactions}.

The performance of DCB networks has been studied in \cite{Gong2011ChannelBounding2,park2011ieee} by means of simulation, and focusing only on a few representative scenarios, thus not allowing for a general characterization of the interactions between neighboring WLANs in DCB networks and their possible effects on throughput and fairness in the resource allocation. In \cite{bellaltachannel}, the performance of a single short-range WLAN using both SCB and DCB is analyzed in the presence of multiple interferers, characterizing the conditions under which DCB outperforms SCB.


This paper is structured as follows. In Section \ref{Sec:SysModel} we introduce the system model. In Section \ref{Sec:Analysis}, we present the analytical model, starting with a toy example and continuing with a general analysis. Section \ref{Sec:Properties} presents intuitions on key properties of the DCB networks that are obtained from the analytical framework. In Section \ref{Sec:11acResults}, we apply our analytical model to IEEE 802.11ac DCB networks and show the accuracy of our analysis by comparing them to simulation results.  


\section{System Model} \label{Sec:SysModel}

Consider $M$ WLANs deployed in an area such that all nodes, regardless of whether or not they belong to the same WLAN, are in the carrier-sensing range of one another.\footnote{This scenario is commonplace as the carrier-sense range is usually more than twice the data communication range~\cite{deng2004tuning}. Therefore, nodes in neighboring WLANs can often sense each other.} \WLANi consists of $U_i$ nodes: an Access Point (AP) and $U_i - 1$ stations (STAs) connected to the AP.

\subsection{WLAN Operational Channel Assignment}

Let $C$ be the entire frequency band available. The frequency band $C$ is divided into $N$ channels of equal bandwidth, hereafter referred to as \textit{basic channels}. A contiguous group of basic channels is assigned to each WLAN during the initialization, which will be used by all the nodes within the WLAN. The channels assigned to different WLANs may or may not have overlaps. We are not concerned with how this assignment is done (internally, externally, randomly, or based on a specific algorithm), but interested in modeling and analyzing the interactions between the WLANs for a given channel assignment. We refer to the channel assigned to \WLANi as $C_i \subseteq C$, and the number of basic channels it contains as $N_i \leq N$. To accommodate for possible limitations imposed by different existing or future protocols, we define a set $\Cs$ from which each $C_i$ has to be chosen. Each member of the set $\Cs$ is a channel consisting in a contiguous group of the basic channels in $C$. The set $\Cs$ is predefined by the channel access protocol. If no limitation is imposed by the protocol, $\Cs$ will contain all possible contiguous groupings of the basic channels in $C$. For the IEEE 802.11ac DCB scheme, the set $\Cs$ is presented in Section \ref{Sec:11ac_channelization}.


Since our goal is to study the interactions between the WLANs, we treat each WLAN as a single entity throughout the paper, unless otherwise stated. We say \WLANi is transmitting if a node operating within it is transmitting. All nodes in all WLANs are assumed to be saturated, i.e., they always have a packet ready for transmission. 




%

\subsection{Generalized Dynamic Channel Bonding (DCB) Scheme}\label{Sec:GenDCB}

In this section, we describe DCB in its most general form. The analytical model that will be presented in Section \ref{Sec:Analysis} can be used to analyze any DCB scheme fitting this general description. However, we will use the specific scheme presented in Section \ref{Sec:P2DCB} for demonstration purposes. 

Under DCB, when \WLANi is initiated, it selects a single basic channel within $C_i$ as its primary channel, and all the other channels are considered as secondary. The primary channel has a distinct role during both backoff and transmission-channel selection, which will become clear in the following subsections. 

\subsubsection{Backoff Procedure}

When a node in \WLANi has a packet ready for transmission, it listens to the WLAN's primary channel. Once the channel has been detected free, the node starts the backoff procedure by selecting a random initial value for the continuous backoff timer, chosen according to a given distribution of mean $1/\lambda$.\footnote{For simplicity of notation, throughout the paper, we focus on a homogeneous case where all WLANs have the same average backoff duration $1/\lambda$. However, our analysis can be applied to heterogeneous cases, by simply replacing $\lambda$ with $\lambda_i$ for each \WLANi.} The node then starts decreasing the backoff value linearly with time, while sensing the primary channel. Whenever a transmission, either from other nodes within the same WLAN or from those belonging to other WLANs, is detected on the primary channel, the countdown\footnote{With a slight abuse of terminology, we will use the word countdown to refer to the continuous decrease of the backoff value.} will be paused until the channel is detected free again, at which point the countdown is resumed. When the backoff timer reaches zero, the node selects a channel for transmission based on the particular channel access scheme under which it is operating (explained in Section \ref{sec:tx_ch_selection}) and starts transmitting.

Once the transmission is finished, the node waits for a predetermined duration to receive an acknowledgement of the proper reception of the transmitted packet. If either the packet or its acknowledgement is corrupted by noise, no acknowledgement is received and the packet has to be retransmitted by starting a new backoff procedure. Otherwise, the entire procedure is repeated for transmitting the next packet. Note that we assume the propagation delay to be negligible. This means that a collision can happen only if two WLANs start transmitting at the same time, which happens with probability zero due to the continuous-time nature of the backoff time. We discuss the implications of this assumption in more details in Section \ref{Sec:Assumptions}.

\subsubsection{Transmission Channel Selection}\label{sec:tx_ch_selection}

At every transmission opportunity of a node in \WLANi, the channel on which the packet will be transmitted is selected based on the status of the basic channels in $C_i$, which are sensed just before the backoff timer reaches zero. The selected channel for transmission, denoted by $c_i$, is the largest contiguous subset of these available channels that contains the primary channel of \WLANi and furthermore belongs to $\Cs$, the set of allowed channels imposed by the particular channel access scheme in use. If there are more than one such largest channels available, one of them will be selected at random. We denote the number of basic channels in $c_i$ by $n_i\leq N_i$. Note that the values of $c_i$ and $n_i$ change at each transmission, whereas $C_i$ and $N_i$ are fixed for all transmissions.

The expected duration of the transmission of a packet, denoted by $1/\mu_{n_i}$, is a function of $n_i$, the number of basic channels used for its transmission. The wider the channel, the shorter the expected transmission duration. The exact relationship between $\mu_{n_i}$ and $n_i$ will depend on the specific technology. In the case of IEEE 802.11ac, this relationship is described in Section \ref{Sec:11acResults}.

\subsection{Two Specific DCB Channelization Schemes}

In this section, we introduce two specific DCB channelization schemes. The first one is used in the IEEE 802.11ac amendment, to which we will refer hereafter as \textit{11acDCB}. The second one is the DCB scheme we use in our examples and demonstrations throughout the paper. We refer to this second scheme as \textit{Powers of Two DCB} or simply \textit{P2DCB}. Both schemes can be explained by simply defining the set of allowed channels $\Cs$.

\subsubsection{11acDCB}\label{Sec:11ac_channelization}

In 11acDCB, $\Cs$ contains only those contiguous subsets of $C$ that are composed of $n=2^k$ basic channels, for some integer $k\leq \log_2 N$, and that their rightmost basic channels fall on multiples of $n$. Figure \ref{Fig:11ac_DCB_channelization} shows $\Cs$ for a scenario where $N = 8$. It contains all the $15$ possible channels containing $\{1,2,4,8\}$ basic channels that also comply with the position limitation described above. Note that, for example, a channel composed of basic channels $\{2,3,4,5\}$ does not belong to $\Cs$ and therefore, is not allowed in the IEEE 802.11ac channelization. 

\begin{figure}[h!]
\centering
\epsfig{file=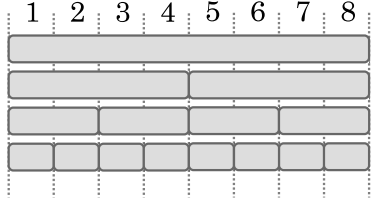,scale=0.6,angle=0}
\caption{The set of allowed channels $\Cs$ in IEEE 802.11ac amendment.}\label{Fig:11ac_DCB_channelization}
\end{figure}

\begin{figure*}[t!]
\centering
\epsfig{file=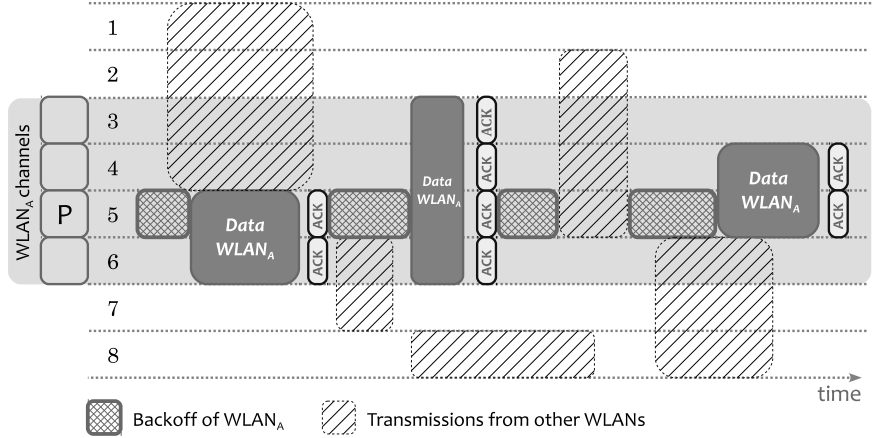,width=0.7\textwidth,angle=0}\label{Fig:ChannelBondingSimple}\\
\caption{Temporal evolution of a WLAN (\WLAN{A}) operating under P2DCB. P indicates the primary channel of \WLAN{A}.}\label{Fig:Schemes}
\end{figure*}

The IEEE 802.11ac amendment also defines a static channel bonding scheme, which is not the focus of this paper, but can nevertheless be analyzed using our analytical model. We refer to this scheme as \textit{11acSCB} and briefly explain it here for completeness. In 11acSCB, the channel access mechanism works exactly as described in our system model, except that the transmission channel $c_i$ can only take a single value $c_i=C_i$. In other words, after the backoff, a node can only transmit if the entire channel assigned to the corresponding WLAN is sensed free. Otherwise, it has to start a new backoff procedure, even if some smaller channel is available. This scheme is clearly inferior compared to the dynamic approach in terms of performance, as also evidenced in Section \ref{Sec:11acResults}.

\subsubsection{P2DCB}\label{Sec:P2DCB}

The P2DCB scheme is simply a DCB scheme with a $\Cs$ that contains all channels $c$ with lengths $n = 2^k$, for some integer $k\leq \log_2 N$. It is similar to 11acDCB in that it only allows for contiguous channels whose length is a power of two, in terms of the number of basic channels they contain. But unlike 11acDCB, there is no limitation on the position of the channel.

As mentioned before, our analytical model is not based on this specific scheme and can be applied to any scheme fitting the generalized DCB description in Section \ref{Sec:GenDCB}. However, we have chosen P2DCB scheme for demonstration purposes due to its similarity to 11acDCB. This similarity enables us to recycle the physical layer channelization parameters in IEEE 802.11ac when calculating the transmission duration for different channel sizes, which results in a fairer comparison of the two schemes.

Figure~\ref{Fig:Schemes} shows the operation of P2DCB scheme. Note that the transmission duration is usually much longer than the backoff duration. However, in order to leave room for clear labeling, the figure is not drawn to scale.



\section{Analysis of DCB Networks} \label{Sec:Analysis}

In this section we put forward a thorough analysis of the behavior of the system described in the previous section. We start with a toy scenario to demonstrate the main concepts and challenges more clearly. Then in Section \ref{Sec:GenAnalysis}, we present a general analysis. Our analysis makes use of continuous-time Markov chains and therefore, requires the backoff and transmission times to be exponentially distributed, which we assume to be the case throughout this section. However, as we will see in Section \ref{Sec:Sensitivity}, DCB networks show little sensitivity to the backoff and transmission time distributions, and therefore, the analytical results obtained using the exponential assumption offer a good approximation for other distributions as well.



\subsection{Toy Example}\label{Sec:ToyExample}

Consider a network consisting of two WLANs, A and B, operating under P2DCB. There are a total of four basic channels available (i.e., $N=4$), and the set of valid channel lengths is $\{1,2,4\}$. The operational channels assigned to \WLAN{A} and \WLAN{B}, $C_\text{A}$ and $C_\text{B}$, are shown in Figure \ref{Fig:ExampleDCB}. \WLAN{A} uses four basic channels, $C_\text{A} = \{1,2,3,4\}$, with channel $2$ assigned as its primary channel. \WLAN{B} only uses channels $3$ and $4$ (i.e., $C_\text{A} = \{3,4\}$ and $N_\text{B} = 2$) and has channel $3$ set as its primary. The primary channels are marked by letter $P$ in the figure. 

We are interested in modeling the overall behavior of the WLANs. The two WLANs in this example can be transmitting simultaneously as long as their transmission channels do not overlap, i.e., $c_\text{A} \cap c_\text{B} = \emptyset$. 

\begin{figure}[h!]
\centering
\epsfig{file=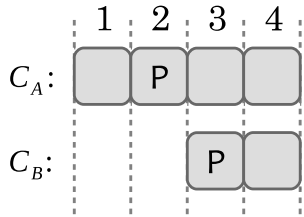,scale=0.75,angle=0}
\caption{Toy example channel assignment}\label{Fig:ExampleDCB}
\end{figure}

\subsubsection{Feasible States}

We define the state of the system at any given time as the set of channels on which each WLAN is transmitting at that time. For the example in Figure \ref{Fig:ExampleDCB}, the set of all feasible states is $\Ss = \{\emptyset, A_4^1, A_2^1, B_2^3, A_2^1 B_2^3\}$, where $\emptyset$ denotes the empty state (when there are no transmissions), and $X_n^j, ~X\in\{A,B\}$, is a shorthand notation that means $c_\Xt = \{j, \ldots, i+n-1\}$, i.e., there is transmission from \WLAN{X} on a channel containing $n$ basic channels, with basic channel $j$ as its leftmost sub-channel. So, for example, when the system is in state $A_2^1 B_2^3$, WLANs A and B are simultaneously transmitting on channels $c_\text{A} = \{1,2\}$ and $c_\text{B}=\{3,4\}$, respectively. Note that there are several other states that may seem feasible at first, but in fact they are not reachable due to the fact that according to DCB, the selected channel for transmission is always the largest available channel in $\Cs$. For example, $A_2^2$ can never be reached because whenever \WLAN{A} finds channel 
$3$ available, \WLAN{B} cannot be transmitting and therefore, channel $4$ will also be available. Since channel $1$ is always available to \WLAN{A}, we can conclude that whenever $A_2^2$ is possible, so is $A_4^1$. So state $A_2^2$ will never be reached.

\subsubsection{Markov Chain Model}
With exponential backoff and transmission times, the aforementioned states form a Continuous-Time Markov Chain (CTMC), as shown in Figure \ref{Fig:ExampleDCB_MC}. To avoid cluttering in the figure, the forward and backward transitions between two states are represented with a single double-sided arrow. The transition rates are indicated by an ordered pair, with the left and right elements corresponding to forward and backward transition rates, respectively. For example, transitions from the empty state to $A_4^1$ happen at rate $U_\text{A}\lambda$, and the transitions in the opposite direction, at rate $\mu_4$. Recall that, $1/\lambda$ is the expected active countdown duration, i.e., excluding the pauses, and is assumed to be the same for all nodes across different WLANs. $U_\Xt$ is the number of nodes in \WLAN{X}. Therefore, the rate at which the collection of nodes in \WLAN{A} access the channel is $U_\text{A}\lambda$. Also, $1/{\mu_{n}}$ is the expected time required for transmission of a packet over a channel consisting of $n$ basic channels. Therefore, backward transitions out of a given state $X_n^j$ happen at rate $\mu_n$.

\begin{figure*}[hhhh!]
\centering
\subfigure[CTMC with transition rates]{\epsfig{file=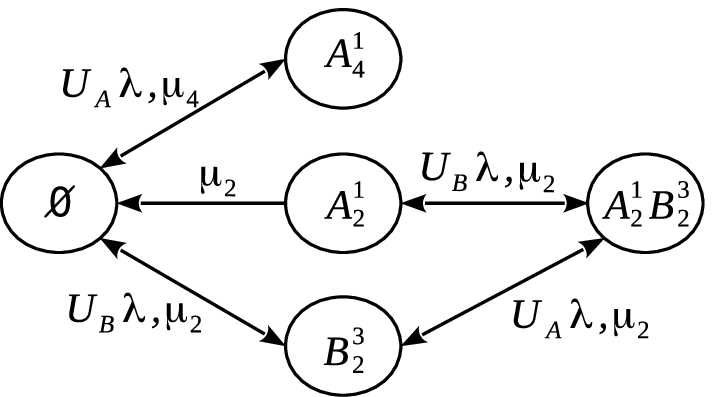,scale=0.9,angle=0}
\label{Fig:ExampleDCB_MC}}
\hspace{1cm}
\subfigure[Order of stat discovery]{\epsfig{file=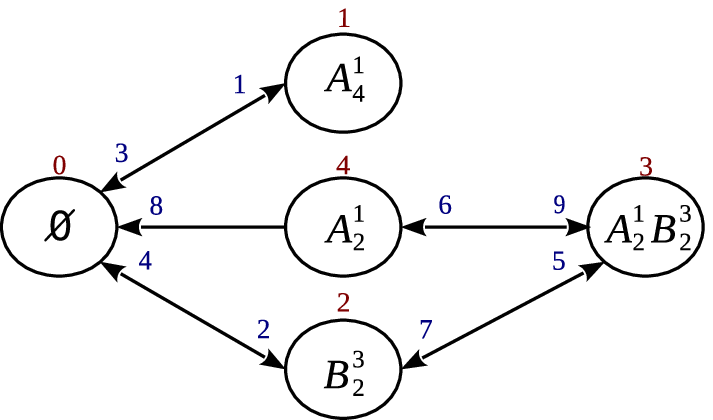,scale=0.9,angle=0}\label{Fig:ExampleDCB_MC_discovery_order}}
\caption{Markov chain corresponding to the toy example in Figure \ref{Fig:4WLAN_Ex_Channels}}
\end{figure*}


There are a few important observations to be made here:
\begin{itemize}
\item Due to the continuous nature of the backoff and transmission times, the probability that two WLANs start or stop transmitting exactly at the same time is zero. Therefore, transitions in the Markov chain can only happen between states which differ by only one WLAN participation.
\item Since always the channel in $\Cs$ with the largest length available is chosen for transmission, some seemingly feasible states are in fact not reachable. 
\item Transitions between two feasible states may in some cases be possible only in one direction. This can be observed for states $\emptyset$ and $A_2^1$, where a forward transition is not valid. This is because at the empty state, all channels are available and if \WLAN{A} is the first one to finish its backoff, then it can use all four available channels for transmission, hence the next state will be $A_4^1$.
\end{itemize}
 
As it can be seen, finding the set of all feasible states and the transitions between them is not trivial. In Section \ref{Sec:MC_construction}, we will propose an algorithm for constructing the corresponding CTMC for any given network of WLANs.


\subsection{General Analysis}\label{Sec:GenAnalysis}


For a system consisting of $M$ WLANs, the state of the system at any given instant can be described by $s = (c_1, \ldots, c_M)$,\footnote{Applying this notation to the toy example, states $A_2^1 B_3^3$ and $A_4^1$ would be equivalent to $s = (\{1,2\},\{3,4\})$ and $s = (\{1,2,3,4\},\emptyset)$, respectively. We will use the two notations interchangeably hereafter.} which indicates the WLANs that are transmitting at that instant and the channels over which they are transmitting. However, as we saw with the toy example, not all states consisting of non-overlapping active channels are in fact reachable. Furthermore, once the reachable states are found, determining valid transitions in the chain may not be trivial. Therefore, we dedicate the next subsection to devising an algorithm that can construct the CTMC corresponding to any given scenario. Once the CTMC is constructed, the stationary distribution of the chain can be calculated and used to calculate different performance metrics of interest, such as throughput and fairness, as done in Section \ref{Sec:Metrics}.

\subsubsection{Constructing the Markov Chain}
\label{Sec:MC_construction}

In order to come up with an algorithm for constructing the CTMC, we first need to understand the basic rules that govern the behavior of a DCB network of multiple WLANs. These basic rules can be summarized as follows:

\begin{enumerate}[{Rule }1.]
	\item The channel selection happens at the beginning of every transmission of every WLAN, and the selected channel will be used until the end of that transmission. 		
	\label{Rule:tx_channel_use}

	\item The selected transmission channel of a given WLAN should always belong to $\Cs$ and contain the WLAN's primary channel. 
	\label{Rule:valid_chan}
	
	\item Two or more WLANs can be transmitting simultaneously only if their transmission channels do not overlap. 
	\label{Rule:no_overlap}
	
	\item When a WLAN wants to start transmitting, it always chooses the largest available channel in compliance with Rule \ref{Rule:valid_chan}. If there are multiple such options, one of them is chosen at random, with equal probabilities. 
	\label{Rule:largest_chan}
	
	\item The probability that two or more WLANs simultaneously finish their backoffs (or equivalently, start their transmissions) is zero.
	\label{Rule:no_simul_start_tx}
\end{enumerate}

Based on these rules, we will now present a few simple rules on the construction of the CTMC for a network of DCB WLANs, in the form of propositions. Let $\Ss$ be the set of all feasible states. Consider two states $s$ and $s' \in \Ss$, such that in $s'$ there is exactly one more WLAN than $s$ transmitting. With a slight abuse of notation, we say $s - s' = X_n^j$, for some \WLAN{X}, to indicate the WLAN that is transmitting in $s'$ but not in $s$ and the channels it is occupying. We always draw the CTMC so that the states with the same number of WLANs participating are grouped into the same column, and those with a larger number of WLAN participation are to the right of those with fewer WLANs. This way, transitions from $s$ to $s'$ and vice versa can be referred to as forward ($s \rightarrow s'$) and backward ($s \leftarrow s'$) transitions, respectively.

\begin{proposition}
Not all valid states in compliance with Rules \ref{Rule:valid_chan} and \ref{Rule:no_overlap} are necessarily reachable.
\end{proposition}
\begin{proof}
A valid state may not be reachable, mainly due to the fact that nodes have to choose the largest available channel in $\Cs$ (Rule \ref{Rule:largest_chan}). The toy example in the previous section can serve as a counterexample to prove this proposition.
\end{proof}

\begin{proposition}
Transitions can only happen between states which differ by only one WLAN participation.
\end{proposition}
\begin{proof}
There are three kinds of transitions that would violate the statement of this proposition: ($i$) those between states that have exactly the same WLANs participating, but with different channels (e.g., $A_2^1 \rightarrow A_4^1$ in the toy example), which cannot happen because it violates Rule \ref{Rule:tx_channel_use}; ($ii$) those between states which differ by more than one WLAN participation (e.g., $\emptyset \rightarrow A_2^1 B_2^3$), which cannot happen due to Rule \ref{Rule:no_simul_start_tx}; and ($iii$) those involving a combination of ($i$) and ($ii$).
\end{proof}

\begin{proposition} \label{Prop:backward_trans}
For any two feasible states, $s, s' \in \Ss$, such that $s' - s = X_n^j$, there is always a transition from $s'$ to $s$, i.e., backward transitions of the form $s\leftarrow s'$ are always feasible. 
\end{proposition}
\begin{proof}
Such transitions happen whenever $X_n^j$ finishes its transmission before the other nodes transmitting in $s'$. This event happens with positive probability, since the transmissions of different WLANs in $s'$ start at random times and last for a random amount of time.
\end{proof}

Consider a state $s\in \Ss$, in which \WLAN{X} is not active. Let $\Ss_{s+X}$ be the set of all possible states that differ from $s$ only by the participation of \WLAN{X}, i.e., those of the form $s' = s + X_n^j$ such that $X_n^j$ is in $\Cs \cap C_X$ and has no overlap with the active channels in $s$. The following proposition characterizes the valid forward transitions from such state $s$.

\begin{proposition}\label{Prop:forward_trans}
There is a transition from $s$ to a given state $\hat{s} = s +  X_{\hat{n}}^{\hat{\jmath}} \in \Ss_{s+X}$ only if $\hat{n} = \max\{n~|~s+X_n^j \in \Ss_{s+X}\}$. 
\end{proposition}
\begin{proof}
This is basically a formalization of Rule \ref{Rule:largest_chan}. Note that there may be more than one such states in $\Ss_{s+X}$ with largest $n$, in which case, a right transition from $s$ to each of those states exists, one of which will be selected randomly at the transition instant. Therefore, the transition rate to such states is multiplied by the inverse of the number of such states. 
\end{proof}

%
Based on these propositions, Algorithm \ref{Alg:CTMC_Construction} presents a systematic way to discover all the feasible states and the valid transitions between them for a given network setup.


\begin{algorithm}[ht!!!]
$i = 0$ \; 
$k = 0$ \;
$s_k = \emptyset$ \;
\While{$s_k \in\{s_0,\cdots, s_i\}$}
{
  $s = s_k$ \; \label{Alg:New_Round}
  \For{every WLAN $X$}{
    \If{$\exists ~ n, j$ such that $X_n^j \in s$ \label{AlgLine:BackwardTransition}}
    {
		$s \rightarrow s - X_n^j$ is a new transition\;
		\If{$s - X_n^j\notin \{s_0,\cdots, s_i\}$}
		{
			$i = i+1$\;
			$s_i = s - X_n^j$\;
      }  
    }
    \ElseIf{$\Ss_{s+X} \neq \emptyset$\label{AlgLine:ForwardTransition}}
    {
		$\hat{n} = \max\{n~|~s+X_n^j \in \Ss_{s+X}\}$\;		
		\For{every $\hat{\jmath}$ such that $X_{\hat{n}}^{\hat{\jmath}} \in \Ss_{s+X}$}
		{ 		
			$s \rightarrow s + X_{\hat{n}}^{\hat{\jmath}}$ is a new transition\;
			\If{$s + X_n^j\notin \{s_0,\cdots, s_i\}$}
			{
				$i = i+1$\;
				$s_i = s + X_n^j$\;
    	  	}
    	 }
    }
    
  }
  $k = k+1$\;
} 
\caption{CTMC Construction. $i$ is the index of the last discovered state, and $k$ is the index of the state currently being used for discovery.}
\label{Alg:CTMC_Construction}
\end{algorithm}

In Algorithm \ref{Alg:CTMC_Construction}, states and transitions are discovered one by one, and every discovered state is then used to discover further states. The empty state, which is always a feasible state for any network configuration, is used in the first round of discovery. In every round of the algorithm, a state $s$ is used as the basis for discovering further state. In the algorithm, variable $i$ is the index of the last discovered state, and $k$ is the index of the state currently being used for discovery (i.e., $s = s_k$). 

As per Proposition \ref{Prop:backward_trans}, for every $X_n^j$ that is active in state $s$, a transition is made into state $s_k - X_n^j$, and the state is added to the list of discovered states if it is not already there. This is done through the \textbf{if} statement of line \ref{AlgLine:BackwardTransition} of Algorithm \ref{Alg:CTMC_Construction}.

Similarly, the forward transitions and their corresponding states are discovered in accordance with Proposition \ref{Prop:forward_trans} through the \textbf{else if} statement of line \ref{AlgLine:ForwardTransition} in the Algorithm, which applies to every WLAN $X$ that is not active in $s$. This is done by first finding the largest channel length $\hat{n}$ available in $\Ss_{s+X}$. Then for all $X_{\hat{n}}^{\hat{\jmath}} \in \Ss_{s+X}$ a new transition from $s$ to $s +X_{\hat{n}}^{\hat{\jmath}}$ is created. If the state $s + X_{\hat{n}}^{\hat{\jmath}}$ was not previously discovered, it is added to the end of the list of discovered states. The algorithm ends when all the discovered states are used for discovery.


Figure \ref{Fig:ExampleDCB_MC_discovery_order} shows the CTMC corresponding to the toy example of Section \ref{Sec:ToyExample}, with the discovery order of transitions and states marked according to Algorithm \ref{Alg:CTMC_Construction}. The first and last discovered states are the empty state and $A_2^1$, respectively. The first discovered transition is $\emptyset \rightarrow A_4^1$ and the last one is $A_2^1 \rightarrow A_2^1 B_2^3$.

\subsection{Solving the CTMC}

Since there are a limited number of possible channel combinations, the constructed CTMC will always be finite. Furthermore, it will be irreducible due to the fact that backward transitions between neighboring states are always feasible. Therefore, a steady-state solution to the CTMC always exists. However, due to the possible existence of the one-way transitions between states (e.g., between $\emptyset$ and $A_2^1$ in the toy example), the CTMC is not always time-reversible and the local balance may not hold \cite{kelly2011reversibility}. The steady-state probabilities of the CTMC can be found by solving the general balance equations, which can be represented in matrix form as
\begin{equation}
\piv \Qmat = 0
\end{equation}
where $\piv$ is the vector of stationary probabilities of the states of the CTMC, and $\Qmat$ is its rate matrix. For the toy example of Section \ref{Sec:ToyExample}, with the states arranged in their discovery order, $\piv = (\pi_\emptyset, \pi_{A_4^1}, \pi_{B_2^3}, \pi_{A_2^1 B_2^3}, \pi_{A_2^1})$, and $\Qmat$ is given by

\begin{small}
\begin{align}
\Qmat = 
\left(\begin{array}{ccccc}
 	-(U_\At + U_\Bt)\lambda & U_\At\lambda & U_\Bt\lambda & 0 & 0\\
    \mu_4     & -\mu_4  & 0 & 0 & 0 \\
	\mu_2 & 0 & 0  & U_\Bt \lambda & -(\mu_2+U_\Bt\lambda)\\
    \mu_2 & 0 & -(\mu_2+U_\At\lambda) & U_\At\lambda & 0 \\
    0 & 0 & \mu_2  & -2\mu_2 & \mu_2\\                    
\end{array}\right) \nonumber
\end{align}
\end{small}


\subsection{Performance Metrics} \label{Sec:Metrics}

Using the stationary distribution, different performance metrics of interest can be found, as detailed in what follows.

\begin{itemize}
    \item \textbf{Throughput:} For a given state $s$, $\pi_s$ is the fraction of time that the network spends in $s$. Let $n_{(X,s)}$ be the channel width that a \WLAN{X}  is using in state $s$, i.e.,
\begin{align}
n_{(X,s)} = \left\{\begin{array}{ll}
n, & X_n^j \in s ~\text{for some $j$}\\
0, & \text{otherwise}\\
\end{array}\right.
\end{align}     
The throughput of \WLAN{X}, in bits per second, is then given by

\begin{align}\label{Eq:WLANThroughput}
    \Gamma_X = L_X \left(\sum_{s \in \Ss}{\mu_{n_{(X,s)}} \pi_s}\right)(1-p_e), 
\end{align}
where $L_X$ is the expected packet length for \WLAN{X}, and $p_e$ is the packet error probability. Furthermore, the network throughput is the sum of the throughputs of all WLANs, and it is given by:

\begin{align}
    \Gamma = \sum_{X}{\Gamma_X}
\end{align}

\item \textbf{Jain's Fairness index}: The Jain's Fairness Index (JFI) with respect to throughput is given by
        \begin{align}\label{Eq:Fairness}
            \field{J} :=\frac{\left(\sum_{i=1}^{M}{\Gamma_i}\right)^2}{\sum_{i=1}^{M}{\Gamma_i}^2}  
        \end{align}        
\end{itemize}



\subsection{Discussion of the Assumptions} \label{Sec:Assumptions}
In this section, we review the assumptions needed for the presented analysis and how they affect the results when our analysis is applied to systems not complying with them.

\subsubsection{The Continuous Backoff Model}

In IEEE 802.11 WLANs, the CSMA/CA procedure uses a slotted backoff. To model CSMA/CA in 802.11, usually Bianchi-like \cite{bianchi2000performance} Markov chain models are used. However, these models work only in scenarios where all participating nodes are able to stay synchronized during their backoff procedure. In systems where the synchronization between different nodes is not possible, e.g., in multihop networks in which some nodes are not in one another's coverage area, continuous-time backoff model is commonly used for modeling the system behavior \cite{laufercapacity}. To do this, the slotted backoff procedure is replaced with a continuous-time backoff with an expected duration equal to that of the slotted system being modeled. These models have been shown to give accurate results in terms of throughput and other performance metrics, matching those obtained through simulations of the slotted version of the system \cite{liew2010back}. 

In our scenario, where multiple neighboring WLANs are operating under DCB scheme, even though nodes are assumed to be within one another's carrier-sensing range, they can still lose synchronization due to the fact that a node only senses its primary channel during backoff, which means that only transmissions on the primary channel result in a pause in the backoff countdown. Therefore, when other nodes transmit on other channels during a target node's backoff, with a slotted backoff, the synchronization between those nodes and the target node would be lost. This is the main motivation for using a continuous-time backoff in our system model.

In Section \ref{Sec:11acResults}, we use our analytical results obtained from the continuous backoff model to characterize the behavior of a CSMA/CA network of IEEE 802.11ac WLANs. By comparing these results to those obtained from simulations, we further validate the accuracy and appropriateness of the continuous backoff models when applied to DCB networks. 


\subsubsection{Zero Propagation Delay}

In our analysis, we assume that there is no propagation delay between the different WLANs in the network. This assumption, together with the continuous backoff assumption, results in a zero collision probability. Without this assumption, we would need to keep track of the WLANs involved in each collision, which would compromise the memoryless property and make the use of Markov chains impossible for our analysis. This can be better explained through a simple example. Consider two WLANs, A and B, operating over the same channel. When A starts a transmission, the system enters state A. However, this transmission can turn into a collision if B starts transmitting within the time it takes for A's signal to reach B (after that, B will not start transmitting because it can already sense A's transmission). If that happens, the system enters state AB, which is the collision state. Otherwise, the system will remain in state A until A's transmission ends and then it goes back to the null state. In the first case, the data that was transmitted while in state A is corrupted, whereas in the second case, it is a successful transmission which counts towards A's throughput. This means that the meaning of state A changes depending on what state follows it. This situation cannot be handled using a Markov chain analysis.

It should also be noted here that, since the WLANs in the network are assumed to be close to each other (all within each other's carrier-sensing range), the propagation delay is quite small and the collision probability is negligibly small. In Section \ref{Sec:11acResults}, we will see that this assumption has little negative impact on the accuracy of the results.

\subsubsection{Saturation}

In our system model, we are assuming that all WLANs are always in saturation condition, i.e., they always have data to transmit. This provides a worst-case scenario in terms of aggressiveness of the WLANs. Our throughput analysis can be easily extended to non-saturated DCB networks by modifying the backoff duration of each WLAN to account for the idle waiting time when the queue is empty, as done in \cite{laufercapacity}. However, in non-saturation conditions, the results are highly dependent on the specific setup, i.e., the traffic load at each WLAN, and it is hard to get intuitions due to an excess of parameters and degrees of freedom. More interesting analysis for non-saturated networks can be obtained by calculating delay and network stability, which is outside the scope of this paper.


\section{Properties of DCB Networks}\label{Sec:Properties}

In this section, we discuss some properties of DCB networks and highlight some of the fundamental differences between DCB networks and the normal single-channel CSMA/CA networks through examples.

In all of the examples presented in this section, the transmission duration when 1, 2, 4 and 8 channels are used is given by $1/\mu_1=12.3$ ms, $1/\mu_2=6.6$ ms, $1/\mu_4=4.6$ ms and $1/\mu_8=3.5$ ms, respectively. In each transmission, $L_d=768$ kbits are transmitted, and the packet error probability is set to $p_e=0.1$. Moreover,  we assume that in each WLAN, only the access point is transmitting (this is equivalent to setting $U_i = 1$ and has no effect on the obtained results).


\subsection{Dominant and Maximal States} \label{Sec:DomMax}


In single-channel CSMA/CA networks, i.e., those in which all nodes operate on a single shared channel with no channel bonding, it is observed \cite{liew2010back} that when backoff times are significantly smaller than transmission times (i.e., $\lambda/\mu_{n} >> 1$ for all feasible channel lengths $n$), the network tends to spend the majority of its time in a small subset of states\footnote{Here a state at any given time is defined by the subset of nodes active at that time.}. In other words, this subset of states have significantly higher steady-state probabilities than the rest of the states. We refer to such states as  \textit{dominant states}. Formally, a state is dominant if its stationary distribution does not vanish as $\lambda/ \mu_{n} \rightarrow \infty$.

The authors in \cite{liew2010back} show that dominant states in single-channel CSMA/CA networks are always those in which the maximum possible number of nodes are simultaneously active, to which we refer as the \textit{maximal states}. Maximal states are easy to pinpoint once the Markov chain of the system is constructed. Since in \cite{liew2010back}, all dominant states are maximal states and vice versa, the maximal states can be used to calculate the throughput of each link, without the need to find the steady-state probabilities, using simple hand computations. 

In the case of DCB networks, a similar dominance behavior is observed, i.e., the system spends the majority of its time in a subset of states. The maximal states are also very easy to find as they are the rightmost states in the constructed Markov chain. However, the dominant states do not always coincide with the maximal states of the chain. In fact, in DCB networks, neither all maximal states are necessarily dominant nor all dominant states maximal.  Therefore, the easy-to-spot maximal states are not useful for finding the dominant states in DCB networks. So even though the dominant states characterize the behavior of the chain and can give intuitions about how the network operates, finding them may require full  calculation of the Markov chain steady-state probabilities, and unlike in single-channel networks, they do not provide a quick method for calculating network performance metrics. 

\begin{figure}[ttth!]
\centering
\epsfig{file=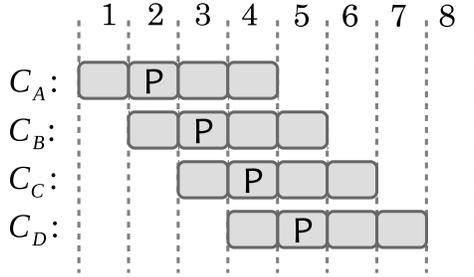
,scale=0.8,angle=0}
\caption{Example with 4 WLANs: channel assignment}\label{Fig:4WLAN_Ex_Channels}
\end{figure}


We show the dominance behavior of DCB networks using an example. Consider a P2DCB network of four WLANs, with channel assignments as shown in Figure \ref{Fig:4WLAN_Ex_Channels}.  The Markov chain corresponding to this network has 49 states as shown in Figure \ref{Fig:4WLAN_Markov}, where the states are depicted by black disks, backward-only transitions by black links, and transitions in both directions by gray links. There are five dominant states in this system, which are those with their state names and corresponding steady-state probabilities printed next to them. The red numbers in parenthesis next to each state indicate their order of discovery and are used for labeling these states in later figures. The dominance of these states can be clearly observed in Figure \ref{Fig:4WLAN_hist}, where the steady-state probabilities of the five dominant states stand out significantly compared to the rest. In fact, the chain spends about $90\%$ of its time in these five dominant states together and the rest of its time in the remaining 44 states. 

\begin{figure}[ttth!]
\centering
\epsfig{file=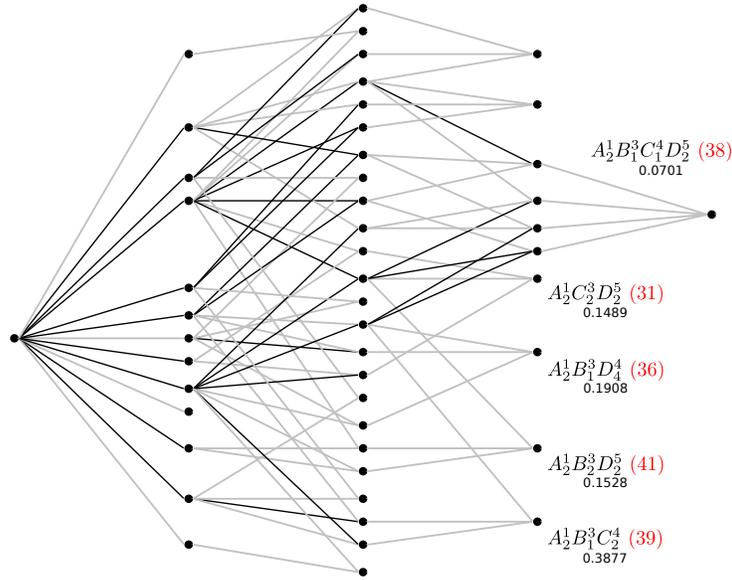,width=0.6\columnwidth,angle=0}
\caption{Example with 4 WLANS (Figure \ref{Fig:4WLAN_Ex_Channels}): Markov chain}\label{Fig:4WLAN_Markov}
\end{figure}

\begin{figure}[ttth!]
\centering
\subfigure[Steady-state probabilities. The states are arranged in their order of discovery.]{\epsfig{file=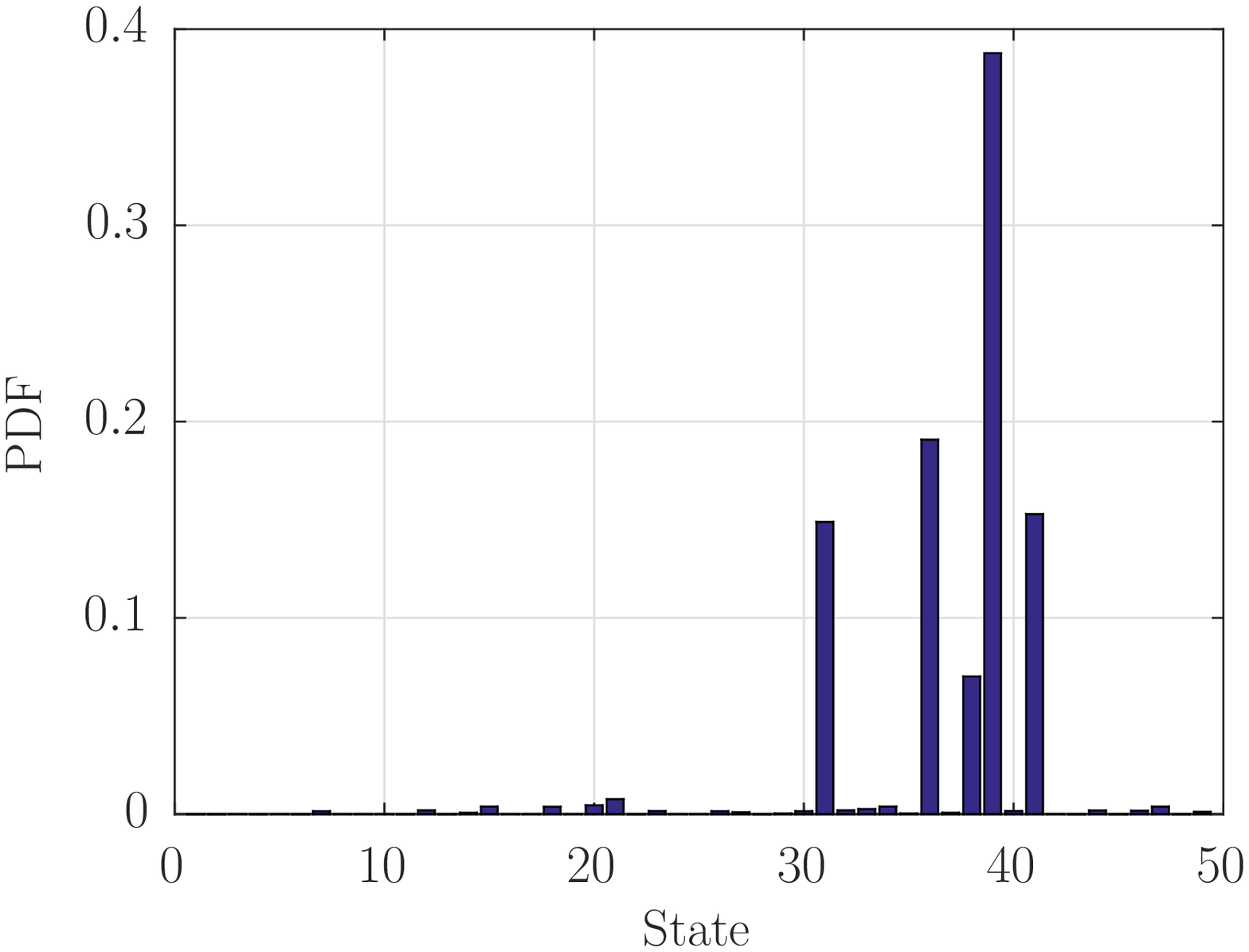,width=0.45\columnwidth,angle=0}
\label{Fig:4WLAN_hist}}
\subfigure[Temporal evolution of the dominant states]{\epsfig{file=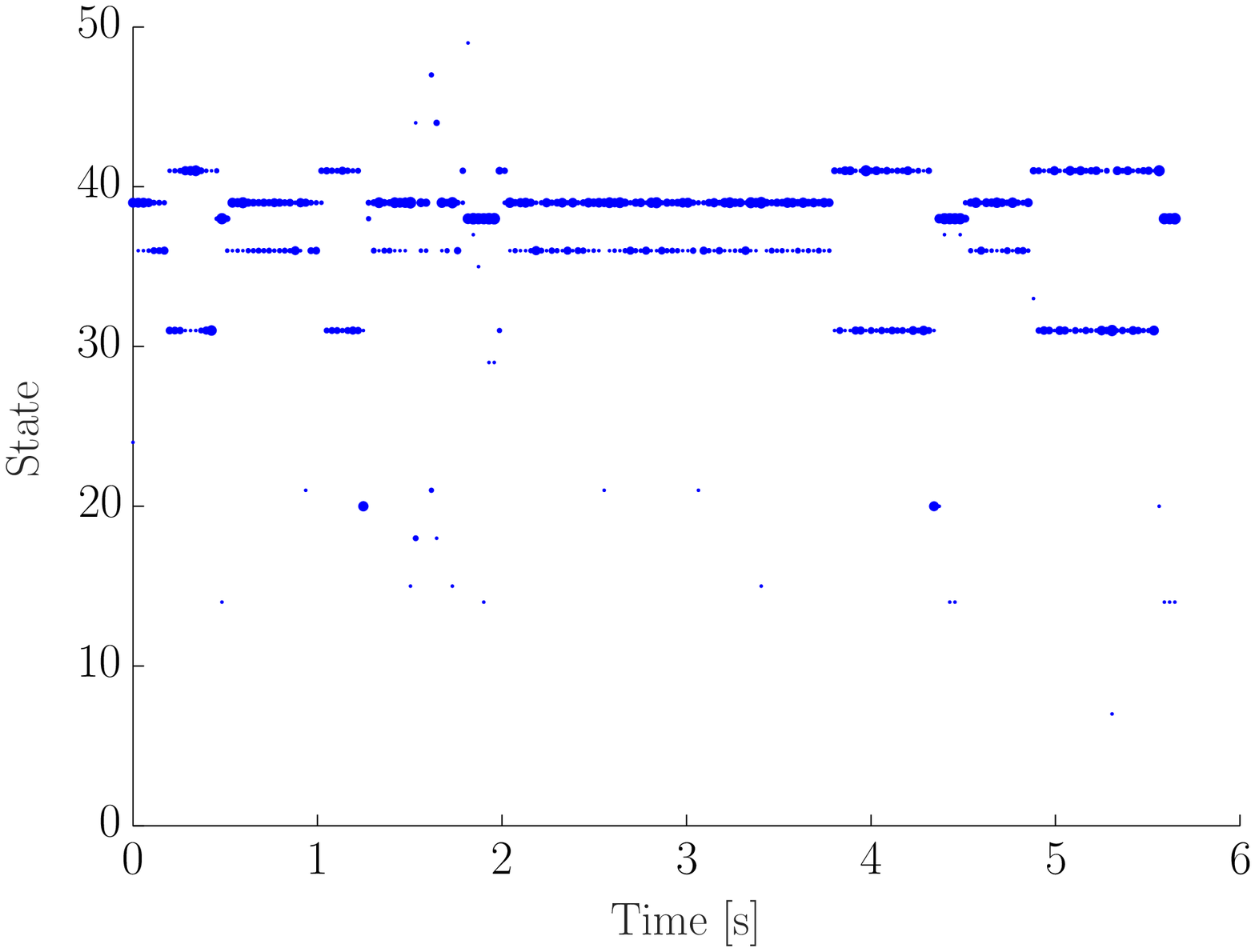,width=0.45\columnwidth,angle=0}
\label{Fig:scatter}}
\caption{Behavior of dominant states in the CTMC of Figure \ref{Fig:4WLAN_Markov} }
\end{figure}

There is a single maximal state in the chain, $A_2^1 B_1^3 C_1^4 D_2^5$, which happens to be the least dominant of the five. However, it is important to observe that all dominant states are \textit{locally maximal} states, namely, states that do not have any neighbor to their right. This observation seems to consistently hold in different scenarios, as long as $\lambda/\mu_{n} >> 1$ for all feasible $n$ values. This behavior can be explained by observing that when  $\lambda/\mu_{n} >> 1$, the right transitions happen at a much faster rate than the left transitions. Therefore, the chain will generally spend the majority of its time in the right half of the chain. Let $s, s'$ be two neighboring states, such that $s\rightarrow s'$ and $s'-s = X_n^j$. Since backoff durations are much shorter than transmission times, it is very likely that while the chain is in state $s$, \WLAN{X} finishes its backoff and the chain enters $s'$ before having spent much time in $s$. Once in $s'$, if any of the other nodes involved in $s'$ finishes its transmission before \WLAN{X}, there will be a backward transition to a state other than $s$. Therefore, the average time spent at $s'$ is generally larger than the time spent in $s$. The locally maximal states enjoy the additional benefit of not having any right neighbors, which means that once entered, the chain will remain there until one of the active nodes finishes its transmission. Note that even though all dominant states are locally maximal, the opposite is not true. 

We summarize the observations in this section as follows:

\begin{observation}
In DCB networks, neither all dominant states are necessarily maximal nor all maximal states are dominant.
\end{observation}

\begin{observation}
In DCB networks, all dominant states are locally maximal states, but the opposite is not necessarily true.
\end{observation}

Although the universal validity of Observation 2 it is not rigorously proven here, it is worth mentioning that we have not been able to construct any example in which it does not hold.

\subsection{Temporal Behavior of Dominant States}

Figure \ref{Fig:scatter} shows the temporal behavior of the dominant states in the example described in the previous section (Figures \ref{Fig:4WLAN_Ex_Channels}, \ref{Fig:4WLAN_Markov}, and \ref{Fig:4WLAN_hist}). In this figure, the time axis is divided into $100$ time windows of \SI{60}{ms} each. For each time window, every visited state is represented by a disc proportional to the time spent in that state during that time window. Therefore, each point in the plot is the result of several transitions between neighboring states. In order to avoid cluttering in the plot, only the top three states visited for the longest time are represented for each time window. 


As can be seen in Figure \ref{Fig:scatter}, the dominant states are the ones in which the Markov chain spends the majority of its time. However, switching between dominant states is not completely random and some of the dominant states are highly correlated in time, i.e.,  consecutive switching between these states is much more likely than to other dominant states. This can be observed in Table \ref{Tbl:SwitchProb}, where the pairwise switching probabilities between different dominant states are listed. In this example, there are three groups of correlated dominant states  $\{\{36,39\},\{31,41\},\{38\}\}$, which are also highlighted in Table \ref{Tbl:SwitchProb} using different colors. 

\begin{table}[th]
  \centering
  \begin{small}
 \begin{tabular}{lccccc}
   \toprule
 &  \textbf{\red{31}} & \textbf{\red{41}} & \textbf{\blue{36}} & \textbf{\blue{39}} & \textbf{38} \\
\textbf{\red{31}}&    0 &   \red{0.9642} &    0.0148 &   0.0014 &   0.0196\\
\textbf{\red{41}}&   \red{0.9621} &        0  &  0.0121 &   0.0012 &   0.0246\\
 \textbf{\blue{36}}&    0.0003 &   0.0094 &        0 &   \blue{0.9899} &    0.0004\\
\textbf{\blue{39}}&    0.0012 &    0.0085 &    \blue{0.9902} &         0 &    0.0002\\
\textbf{38}&    0.1545 &    0.1218 &    0.3501 &    0.3735 &         0\\
  \bottomrule
 \end{tabular}
  \end{small}
 \caption{Switching probabilities for dominant states in the Markov chain of Figure \ref{Fig:4WLAN_Markov}.}\label{Tbl:SwitchProb}
\end{table}

It can also be noted in Figure \ref{Fig:scatter} that once the chain enters one of these groups, it stays within that group for a while. This behavior is observed even when the group contains only one state. This is because every dominant state is surrounded by a group of non-dominant states from which transitions are most likely to that dominant state, which is a consequence of the fact that dominant states are locally maximal states.  Table \ref{Tbl:SojRet} shows the median\footnote{We use the median instead of the mean to reduce the effect of rare outliers on the results.} sojourn times and return times for the groups of dominant states in the 4-WLAN example. The sojourn time here is defined as the time between the moment the chain enters one of the dominant states in a given group until the first time it enters a dominant state from another group. Similarly, the return time is the time it takes for the chain to return to a dominant group after it exits that group. We can see in this table that state $38$, which is the least dominant state has a significantly higher return time than the other groups of dominant states. 
\begin{table}[th]
  \centering
  \begin{small}
  \begin{tabular}{lcc}
   \toprule
  {\bf States} & {\bf Sojourn time} & {\bf Return time} \\
  \midrule
  $\{31,41\}$ & 0.2591 & 0.4356 \\
  $\{36,39\}$ & 0.4249 &  0.3556\\
  $\{38\}$ & 0.1509 & 2.0188\\
  \bottomrule
 \end{tabular}
  \end{small}
 \caption{Sojourn and return times for dominant states in the CTMC  of Figure\ref{Fig:4WLAN_Markov}.}\label{Tbl:SojRet}
\end{table}
In terms of system performance, at any given time, the achieved throughput depends on the group of states in which the system is. Since the chain tends to stay in dominant groups for a long time, this may result in long delays for networks that are not present in that dominant group. Nevertheless, for this specific example  severe starvation effects are not observed because each dominant group has all WLANs transmitting in at least one of its dominant states. But it is important to notice that this may not always be the case, and the transient behavior between dominant states needs to be studied to understand the delay properties of the network.

It should be noted here that, to observe the temporal behavior explained above, the time window duration in Figure \ref{Fig:scatter} has to be chosen carefully. If the time window is too long, all dominant states will group together and the end result will only represent the steady-state probabilities of the dominant states, i.e., the same information as in Figure \ref{Fig:4WLAN_hist}, but presented differently. If the time window is too short, then during each time window, only one state will be visited. So the plot will be very cluttered, simply showing a sample path of the Markov chain, which means that the dominance behavior cannot be properly observed. 
In order to choose a reasonable value for the time window, we first calculate the \textit{mixing time} of the Markov chain. The mixing time, $T_\epsilon$, represents the time at which the steady-state distribution is within an $\epsilon$ radius of $\piv$. Roughly speaking, this is the time at which the state of the chain becomes independent of the state in which it started. Following~\cite{goel2006mixing}, we can get an upper bound $\hat{T}_\epsilon$ for the $L^2$ mixing time $T_\epsilon$. Considering the scenario of Figure~\ref{Fig:4WLAN_Ex_Channels}, for $\epsilon=0.001$, $\hat{T}_\epsilon$ is about \SI{1.45}{s}. In order to see a complete evolution of the Markov chain, Figure \ref{Fig:scatter} is plotted for \SI{6}{seconds}, i.e., approximately $4\times\hat{T}_\epsilon$. This time is then divided into 100 time windows.




\subsection{Sensitivity to the Backoff and Transmission Time Distributions} \label{Sec:Sensitivity}

One of the key properties of single-channel CSMA/CA networks with continuous-time backoff is that the stationary distribution of the Markov chain is insensitive to both the backoff and the packet size distributions \cite{liew2010back,van2010insensitivity}. This means that, when analyzing such networks using Markov chain models, even though constructing the Markov chain requires both the backoff and transmission times to be exponentially distributed, the different performance metrics calculated using this Markov chain analysis are valid for any backoff and transmission time distribution with the same expected value. The same result is also observed in SCB networks \cite{bellalta2015interactions}. 

It is shown in \cite{liew2010back}, that the Markov chains corresponding to single-channel CSMA/CA networks are time-reversible, and that this is a sufficient condition for these networks to be insensitive to backoff and transmission time distributions. DCB networks, however, often have Markov chains that are not time-reversible. For example, the presence of one-way transitions in the Markov chain is sufficient for it not being time-reversible. Therefore, the results in \cite{liew2010back} do not apply to DCB networks and, as corroborated by simulations, these networks may show different degrees of sensitivity to backoff and transmission time distributions.

\begin{figure*}[ttth!]
\centering
\subfigure[Throughput of E/E case]{\epsfig{file=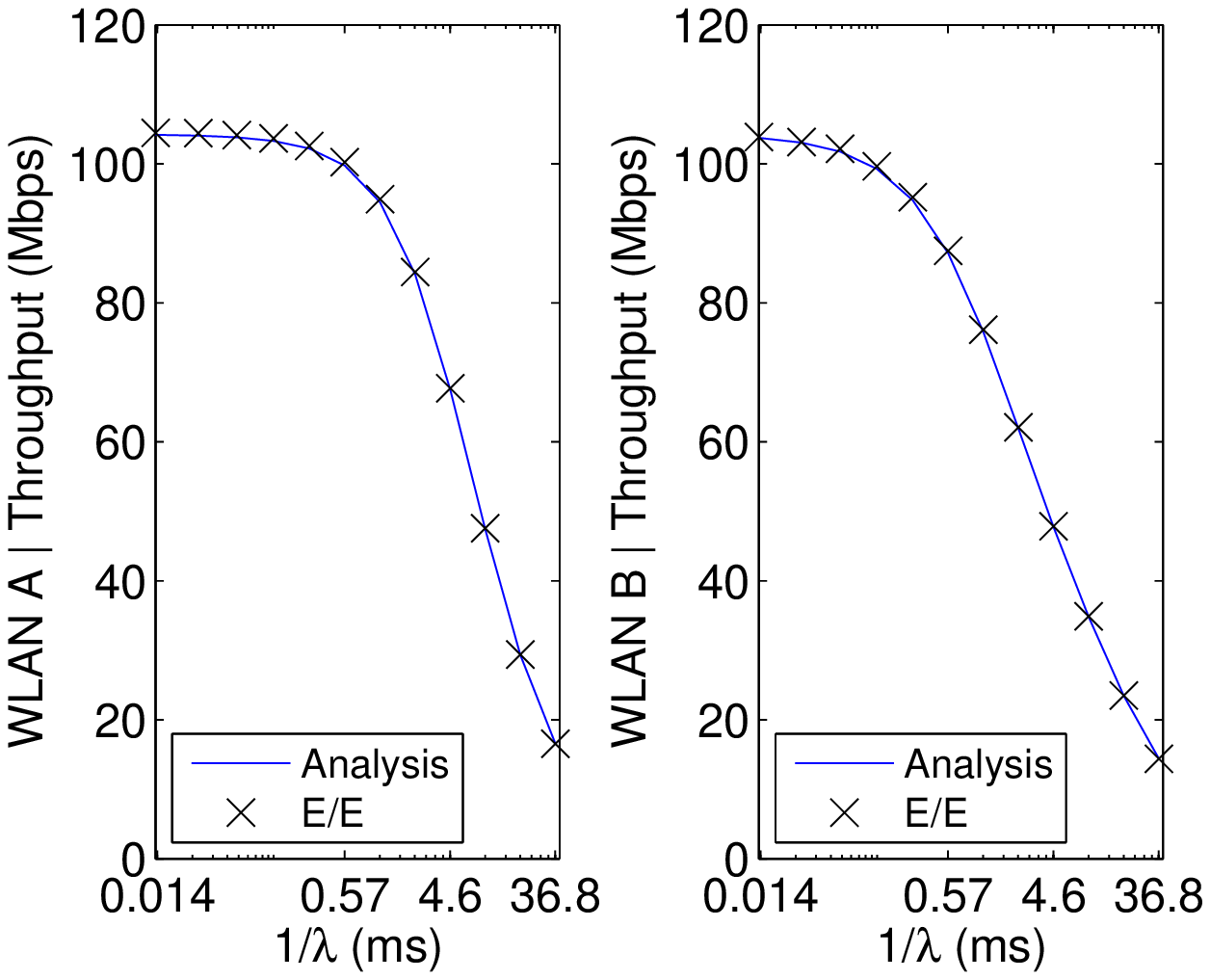,scale=0.35,angle=0}\label{Fig:SensitivityTwoWLANsAB_Throughput}}
\subfigure[Throughput difference with E/E case]{\epsfig{file=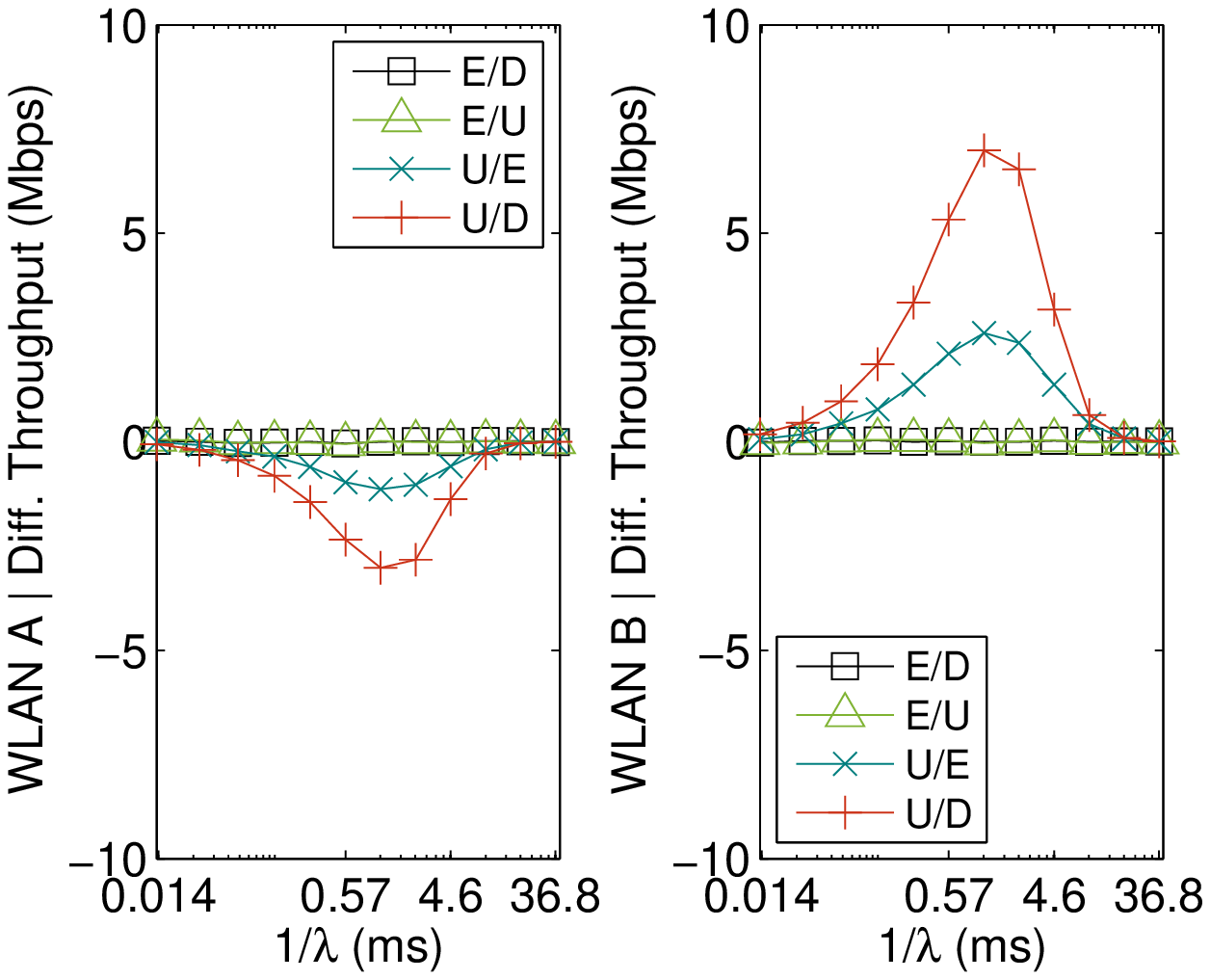,scale=0.35,angle=0}\label{Fig:SensitivityTwoWLANsAB_Dif}}
        \subfigure[Expected channel width of \WLAN{A} (in number of basic channels)]{\epsfig{file=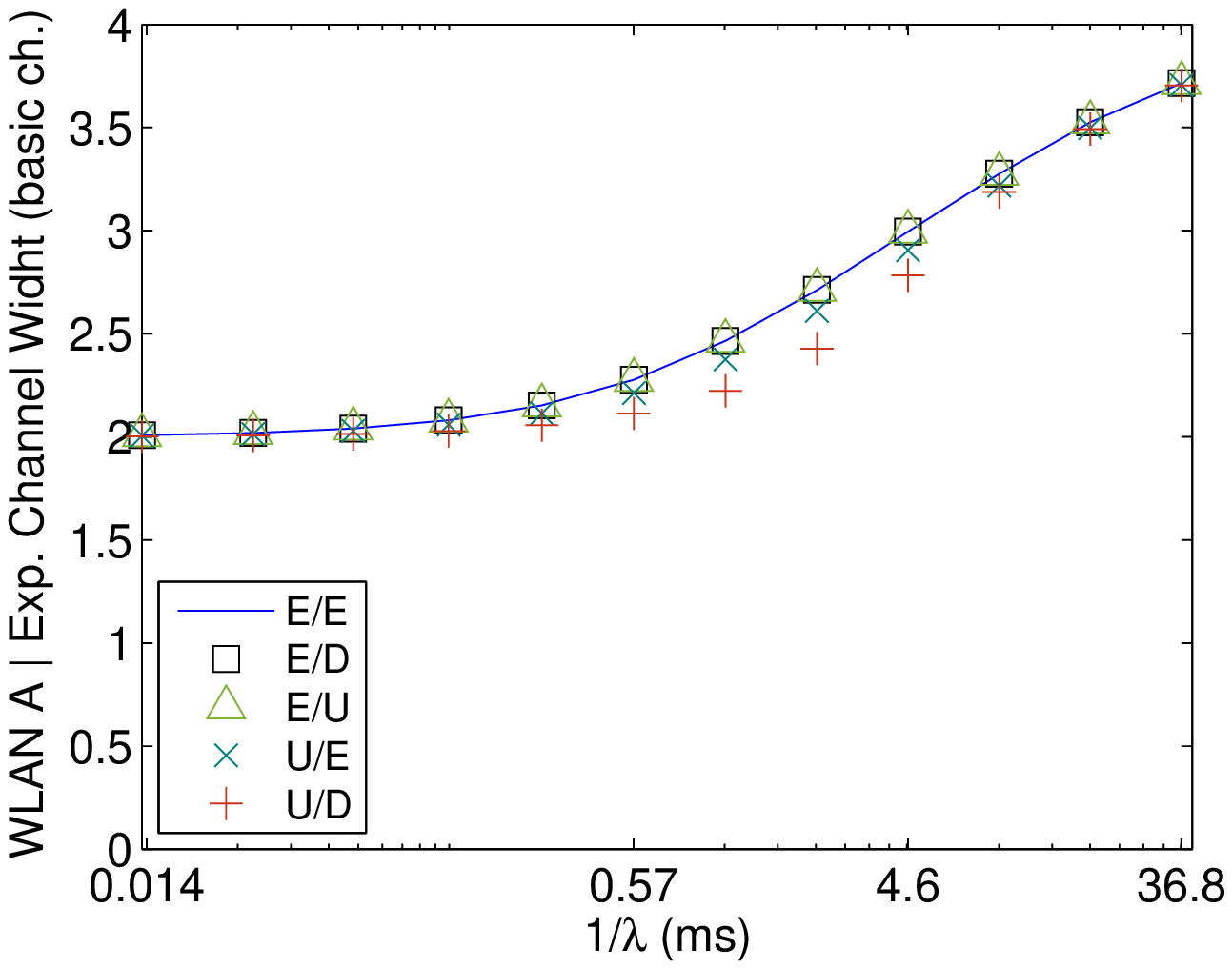,scale=0.35,angle=0}\label{Fig:SensitivityTwoWLANsAB_CWidth}}
\caption{Evaluation of the sensitivity to the backoff and transmission durations of the CTMC of Figure \ref{Fig:ExampleDCB_MC} (toy example).}\label{Fig:Insensitivity}
\end{figure*}


Figure \ref{Fig:Insensitivity} shows the throughput sensitivity of each WLAN in the toy example of Figure \ref{Fig:ExampleDCB} to the backoff and transmission time distributions. In Figure \ref{Fig:SensitivityTwoWLANsAB_Throughput}, the network throughput is plotted for the case where both the backoff and transmission times are exponentially distributed. This serves as a reference for Figure \ref{Fig:SensitivityTwoWLANsAB_Dif}, where the difference between the throughput achieved using different backoff and transmission time distributions and that achieved using exponential distribution is depicted. In the legend, the backoff and transmission time distributions are indicated by their first letter with E, U, and D referring to exponential, uniform, and deterministic, respectively. For example, E/U refers to exponential backoff and uniform transmission time. 

The sensitivity observed in Figure \ref{Fig:SensitivityTwoWLANsAB_Dif} can be explained by the expected channel width of \WLAN{A}, as shown in Figure \ref{Fig:SensitivityTwoWLANsAB_CWidth}. In the case of U/D where the highest sensitivity is observed, \WLAN{A} spends less time using the four basic channels compared to the E/E case, which results in a lower throughput for \WLAN{A}. 

In Section \ref{Sec:11acResults}, we further compare our analytical results to those obtained from simulating DCB networks of IEEE 802.11ac WLANs, in which the backoff and transmission times are not exponential. In that and other cases that we have simulated, the observed sensitivity is relatively small and the analytical results provide a good approximation for the throughput.



\section{Analysis of IEEE 802.11 WLANs} \label{Sec:11acResults}

In this section, we apply our analysis to DCB networks of IEEE 802.11 WLANs. These networks differ from the DCB network model considered in our analysis mainly in the following two aspects:
\begin{enumerate}
\item IEEE 802.11 uses a slotted backoff mechanism.
\item The backoff and transmission times in IEEE 802.11 are not exponentially distributed.
\end{enumerate}
Despite these differences, we will see that our analytical model gives accurate results when applied to DCB networks of IEEE 802.11 WLANs. 

We also compare the channelization used for channel bonding in IEEE 802.11ac amendment, 11acDCB, to P2DCB, the channelization used in all previous examples, and see that the limitations imposed by 11acDCB channelization can have a negative effect on the performance of DCB networks.


The results presented in this section are obtained for a group of neighboring IEEE 802.11 WLANs using DCB. In order to focus only on the interactions between different WLANs, we assume that each WLAN is only transmitting downlink traffic (i.e., only the access point is transmitting). All WLANs operate in the $5$ GHz ISM band, where each basic channel has a width of $20$ MHz. Therefore, the operating channel $C_i$, selected by \WLANi can  have a bandwidth of $\{ 20, 40, 80, 160\}$ MHz. 

In Table \ref{Tbl:Parameters80211ac}, the main parameters used in this section and their values are presented.

\begin{table}[thhhhhhhhhhhhhhh!!]
\centering
\begin{small}
 \begin{tabular}{lcc}
   \toprule
  {\bf Parameter} & {\bf Notation} & {\bf Value} \\
  \midrule
  Packet length & $L_d$ & $12000$ bits \\ 
  Number of aggregated packets & $K_{\text{A}}$ & $64$ packets \\ 
  Contention window & CW & $16$ slots \\ 
  Slot duration & $T_{\text{slot}}$ & $9~\mu$s \\
  Packet Error Probability & $p_e$ & $0.1$ \\ 
  \bottomrule
 \end{tabular}
\end{small}
 \caption{Parameters considered to obtain the results presented in Section \ref{Sec:11acResults}.}\label{Tbl:Parameters80211ac}
\end{table} 

Table \ref{Tbl:RateAdaptation}, shows the modulation $K_\text{m}$ and the coding rate $R$ for each channel width $n$. $\xi(n)$ is the number of data subcarriers when $n$ basic channels are bonded together. The modulation and coding rate are adapted to compensate for the change in signal-to-noise ratio. The values presented in Table \ref{Tbl:RateAdaptation}, are the reference values given by the IEEE 802.11ac amendment to keep the error probability $p_e$ below $10\%$. We use these values to calculate the corresponding packet transmission duration, $\frac{1}{\mu_n}$, for a given channel width, $n$, as indicated in the rightmost column of Table \ref{Tbl:RateAdaptation}, and use a fixed error probability of $p_e = 0.1$ in both the simulation and the analytical calculations. The calculation of $\frac{1}{\mu_n}$ is detailed in Appendix. 

\begin{table}[thhhhhhhhhhhhhhh!!]
\centering
\begin{small}
 \begin{tabular}{lcccc}
  {\bf $n$} & {\bf $\xi(n)$} & {\bf $K_{\text{m}}$} & {\bf \textbf{$R$}} & \bf{$\frac{1}{\mu_n}$}\\
  \midrule
  1 & $52$ & $6$ bits (64-QAM) & $5/6$ & 12.26 ms\\
  2 & $108$ & $6$ bits (64-QAM) & $3/4$ & 6.63 ms\\
  4 & $234$ & $4$ bits (16-QAM) & $3/4$ & 4.64 ms\\
  8 & $468$ & $4$ bits (16-QAM) & $1/2$ & 3.52 ms\\  
  \bottomrule
 \end{tabular}
 \end{small}
 \caption{Modulation and coding rates used for different channel width.}\label{Tbl:RateAdaptation}
\end{table} 


The simulator is based on the COST (Component Oriented Simulation Toolkit) libraries \cite{chen2002reusing}. It accurately reproduces the described scenarios and the operation of each node.


\subsection{Slotted Backoff in IEEE 802.11 WLANs}

\begin{figure}[ht!]
\centering
\epsfig{file=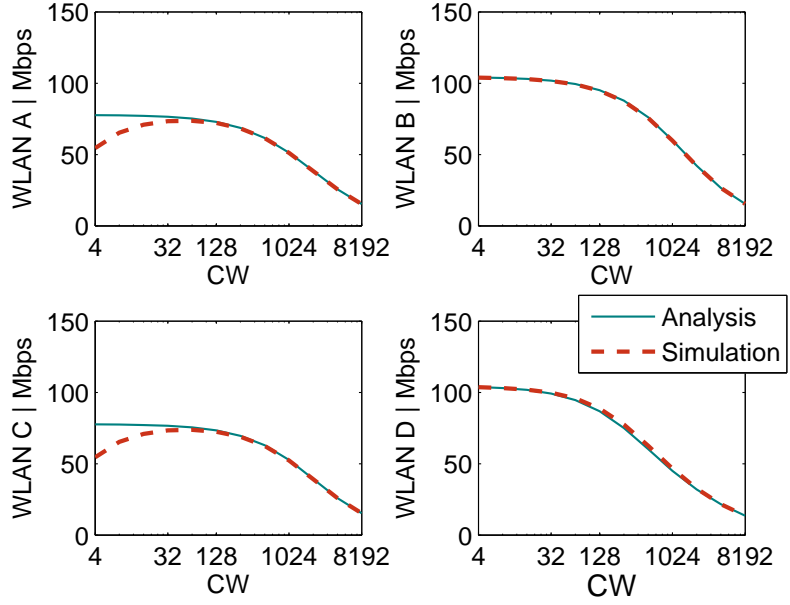, width = 0.7\columnwidth, angle=0}
\caption{Comparative analysis between the throughput obtained from the model with the ones considering the slotted backoff mechanism in IEEE 802.11. The relationship between $\lambda$ and CW is given by $\lambda=\frac{2}{(\text{CW}-1)T_\text{slot}}$.}\label{Fig:ModelValDCF}
\end{figure}

The backoff process in IEEE 802.11 WLANs is implemented using a discrete counter which is decremented every $T_{\text{slot}}$ seconds when the channel is detected  empty. Every node picks an integer backoff value in the range [0,CW-1], and decreases the counter. When the counter reaches zero, the node starts transmitting the packet on the selected channel. When several nodes finish their backoff at the same time, all transmitted packets collide. We assume that all packets involved in a collision are corrupted and lost.

Figure \ref{Fig:ModelValDCF} shows the throughput achieved by a group of $4$ WLANs (A, B, C, and D) using the following initial channel assignment: $C_\text{A} = \{1,\ldots,8\},~ C_\text{B} = \{1,\ldots, 4\},~C_\text{C} = \{5,\ldots, 8\},~C_\text{D} = \{1,2\}$, with their primaries on basic channels $5,3,7$, and $1$, respectively. It can be seen in the figure that for small CW values, the throughput obtained analytically is higher than the one obtained by simulations. This is because in the analytical model, a continuous backoff model is used, which results in a zero probability of collision. With slotted backoff, however, this probability is non-zero. When CW is small, the backoff duration is shorter and therefore the collision probability is higher. This results in a lower throughput for small CW. For higher values of CW, the analytical and simulated throughput values shows a very good match. The effect of collisions is much more noticeable for WLANs A and C. This is because the two dominant states in this chain are $A_4^5 B_2^3 D_2^1$ and $C_4^5 B_2^3 D_2^1$, which means that WLANs A and C are contending for channel $\{5, \ldots, 8\}$ and therefore, may suffer collisions, while WLANs B and D are transmitting in two independent channels most of the time. Note that in all WLANs, for the optimal CW value, i.e., the CW value that maximizes the throughput, the analysis and the simulations match closely. 

The accuracy of these results is also another demonstration of the low sensitivity of DCB networks to backoff and transmission time distributions; in the simulated network, neither backoff nor the transmission time are exponentially distributed while the analytical results are obtained assuming exponential distribution for both. 



\subsection{IEEE 802.11ac Channelization}

\begin{figure}[t!]
\centering
\subfigure[]{\epsfig{file=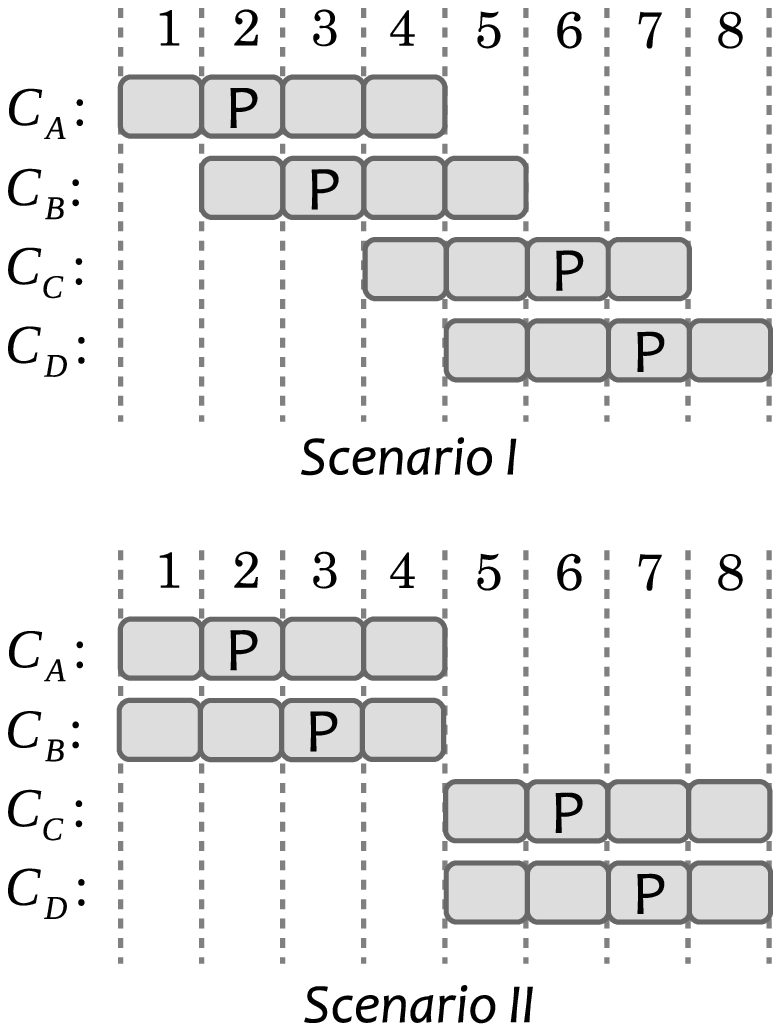,width=0.35\columnwidth,angle=0}\label{Fig:11ac_channels}}\quad\quad
\subfigure[]{\epsfig{file=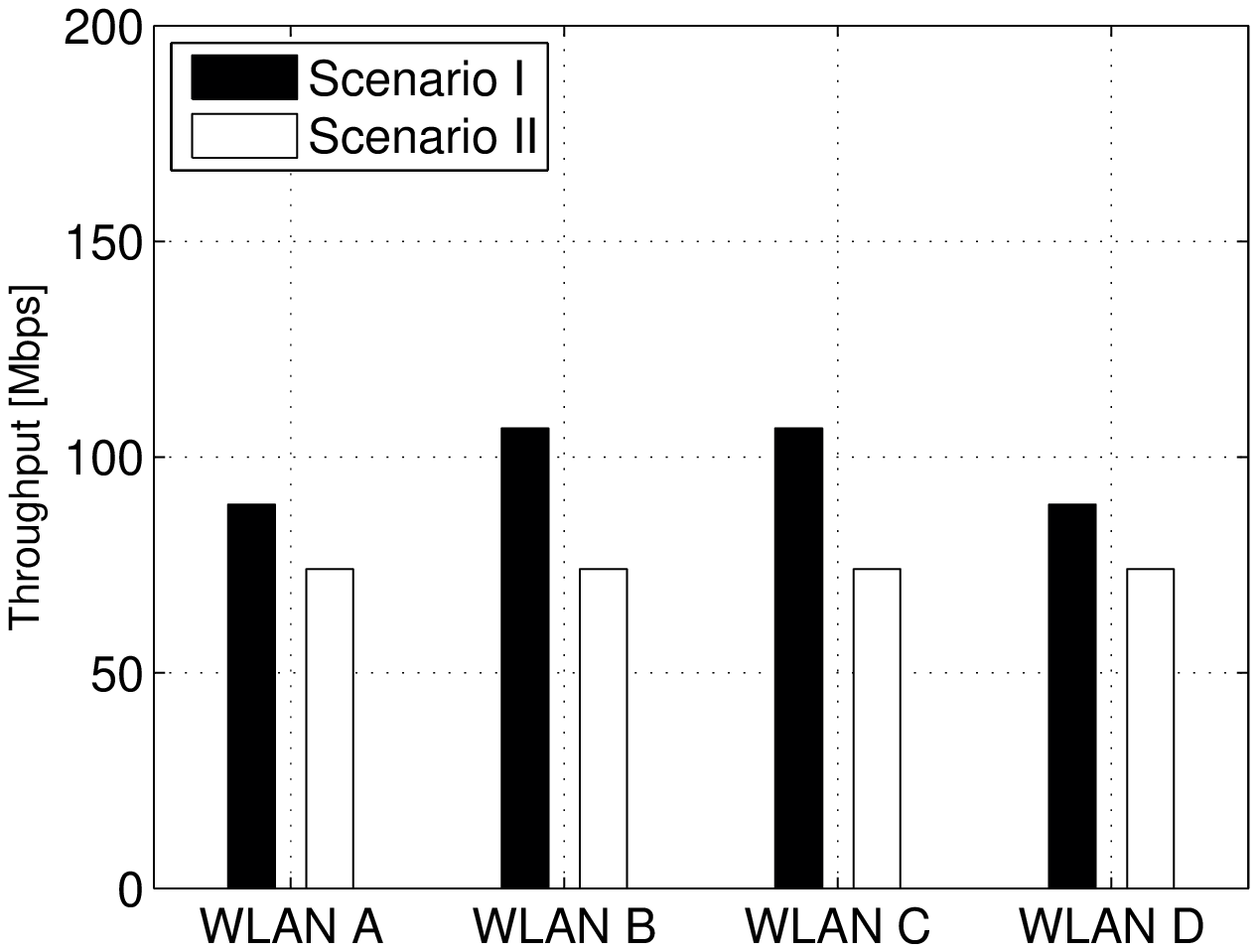,width=0.5\columnwidth,angle=0}\label{Fig:11ac_throughput}}

\caption{Scenario I operates under P2DCB, and II under 11acDCB.}\label{Fig:11ac_channelization}
\end{figure}

The IEEE 802.11ac channelization (Figure \ref{Fig:11ac_DCB_channelization}) has been designed to prevent partial overlaps between channels of the same width. It was shown \cite{bellalta2015interactions} that such channelization improves the performance when static channel bonding (SCB) is used. This is due to the fact that in SCB networks, for a \WLANi to start a transmission, all basic channels in $C_i$ must be found empty, and therefore a WLAN that partially overlaps with several WLANs is blocking all of them.  

To test the impact of the IEEE 802.11ac channelization in DCB networks, we compare the two scenarios depicted in Figure \ref{Fig:11ac_channels}. In Scenario I, there is partial overlap between channels of the same size and therefore it is not compatible with 11acDCB channelization. This scenario operates under P2DCB. Scenario II operates under 11acDCB. Since partial overlaps between channels of the same size are not allowed, the WLANs are forced to have fully overlapping channels if they all want to use channels containing four basic channels. The analytically obtained throughput for the four WLANs in both scenarios is plotted in Figure \ref{Fig:11ac_throughput}. As can be seen in this figure, the limitation imposed by 11acDCB can result in a lower throughput. 




\subsection{DCB vs SCB}

\begin{figure*}[ttth!]
    \centering
\subfigure[Throughput]{\epsfig{file=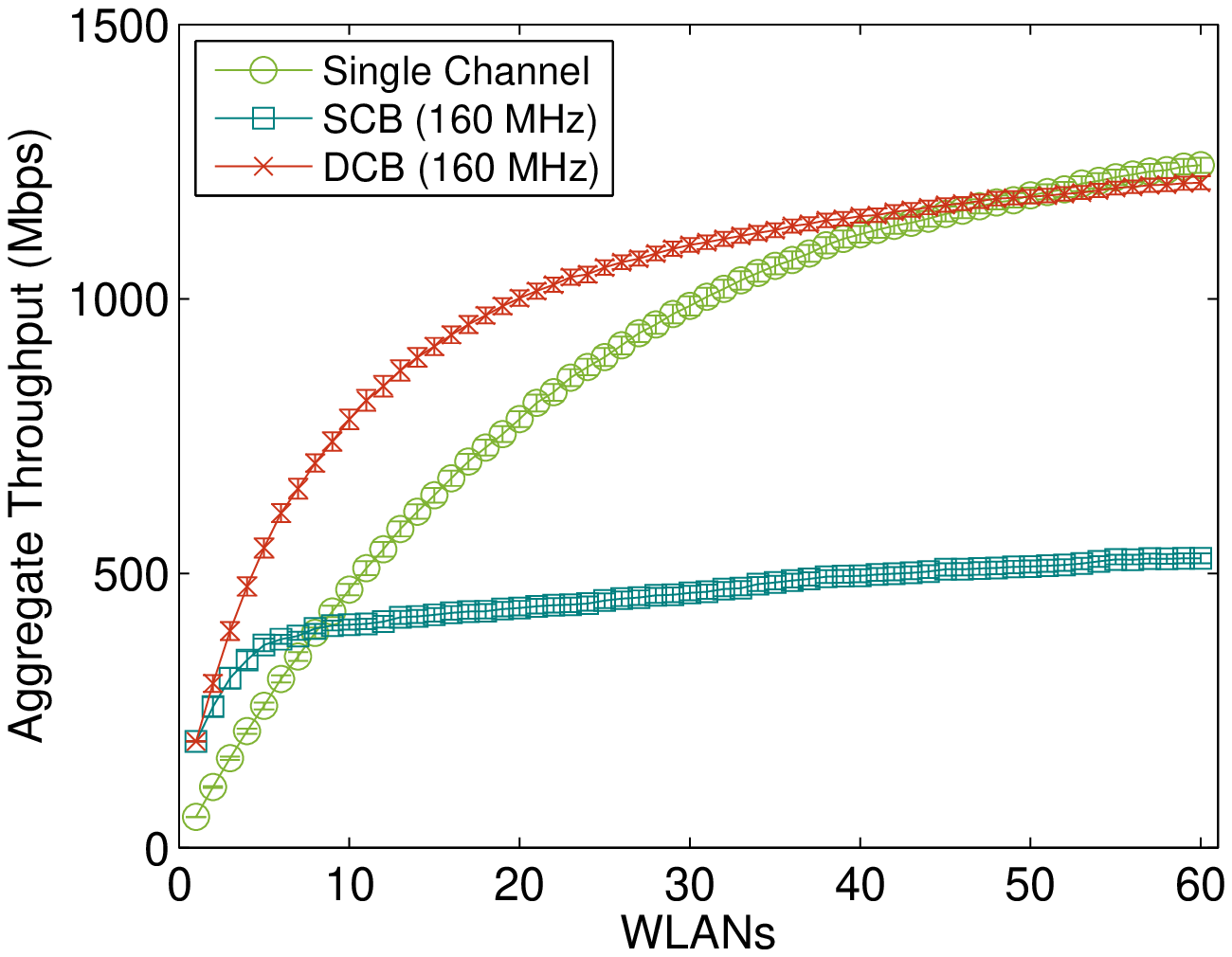,scale=0.35,angle=0}\label{Fig:FigureIncresingWLANs_S}}
\subfigure[Jain's Fairness Index]{\epsfig{file=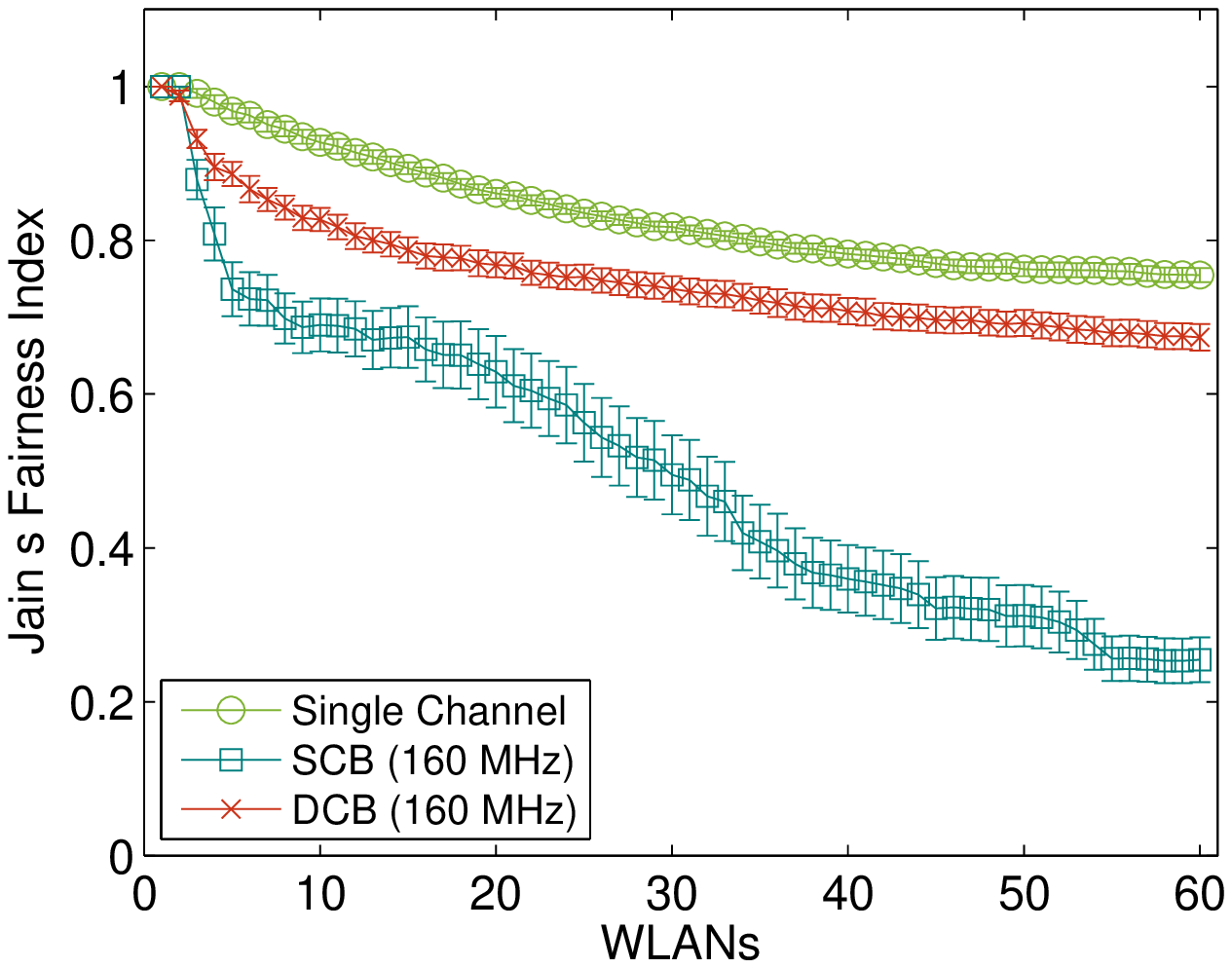,scale=0.35,angle=0}\label{Fig:FigureIncresingWLANs_J}}
\subfigure[Expected number of basic channels per transmission]{\epsfig{file=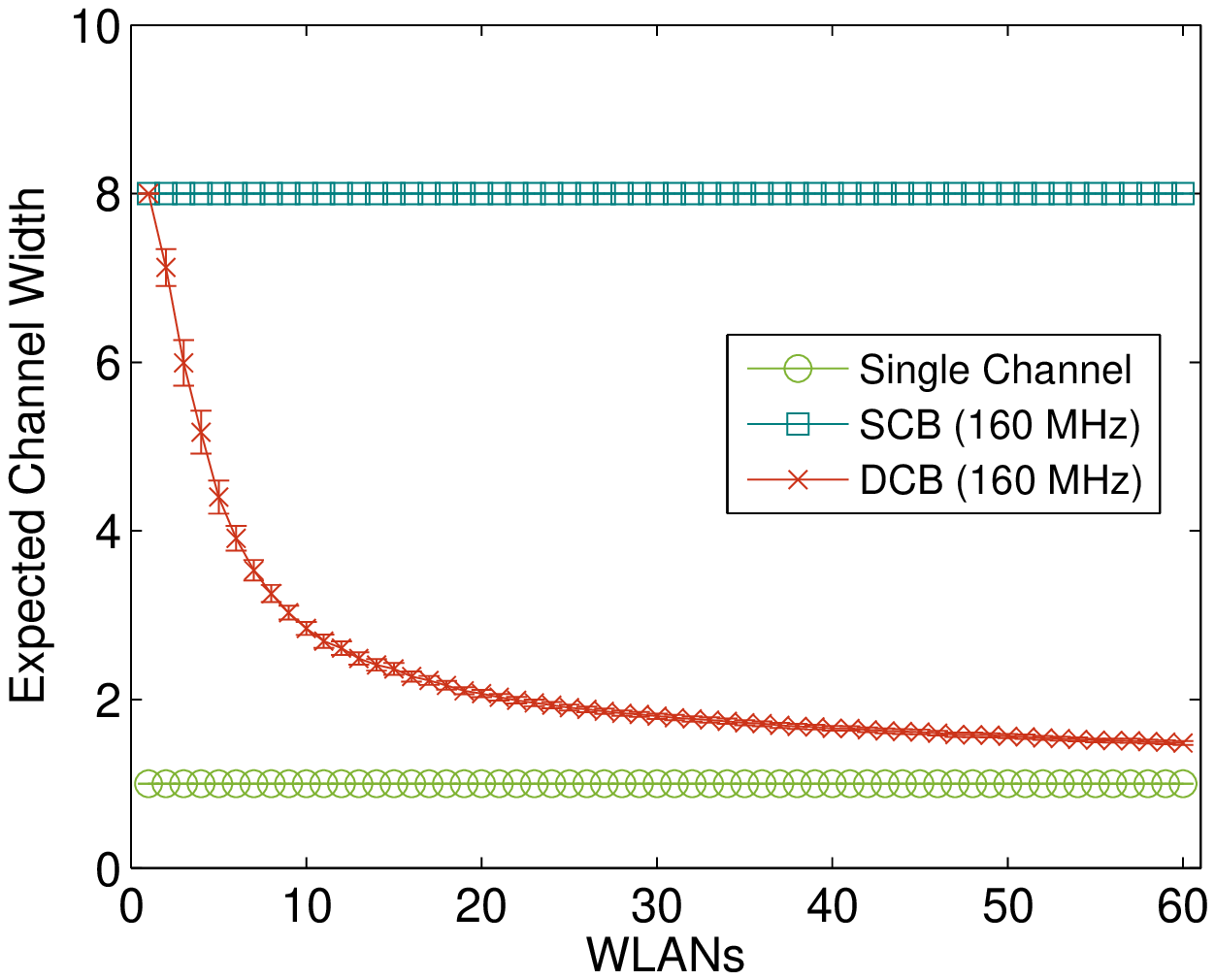,scale=0.35,angle=0}\label{Fig:FigureIncresingWLANs_EW}}
\caption{Aggregate throughput, Jain fairness index of the throughput, and expected number of basic channels used at each transmission when the number of neighboring WLANs increases. The error bars represent the 95 \% confidence intervals.}
\label{Fig:Figure1_S}
\end{figure*}

In this section we investigate the throughput achieved by DCB networks in dense WLAN scenarios using both static and dynamic channel bonding. Figure \ref{Fig:Figure1_S} shows different performance metrics for a group of neighboring WLANs when their number increases from 1 to 40. Each point in each of the  figures is the result of averaging 100 simulations of 100 seconds of duration each. The total number of available basic channels is set to 24. The initial channel assignment is done randomly as follows. For a \WLANi, the length of the operational channel is set to $N_i=8$. Then the location for the leftmost basic channel in $C_i$ is selected randomly from $\{1,\ldots, 17\}$. One of the 8 basic channels in $C_i$ is randomly selected as the primary channel for \WLANi. For both DCB and SCB, the P2DCB channelization is used when selecting the transmission channel $c_i$. In all cases the backoff and the packet transmission duration are exponentially distributed. 

The aggregate throughput, the Jain's Fairness Index, and the expected channel width used at every transmission by each WLAN are shown in Figure \ref{Fig:Figure1_S}. Figure \ref{Fig:FigureIncresingWLANs_S} shows that, on average, DCB outperforms SCB in terms of aggregate throughput, providing also a fairer throughput distribution between WLANs (Figure \ref{Fig:FigureIncresingWLANs_J}), for all network sizes, due to its ability to adapt the channel width to the channel congestion level (Figure \ref{Fig:FigureIncresingWLANs_EW}). When the density of WLANs is very high, the use of a single channel by all WLANs is the best option as it minimizes the contention between the WLANs and the expected channel width in DCB  converges to one. Note that the achieved throughput in all three cases can be improved by using an optimum channel allocation instead of the random allocation used for plotting Figure \ref{Fig:Figure1_S}.



\section{Conclusions} \label{Sec:Conclusions}

In this paper, we introduced an analytical framework, able to capture the interactions between neighboring WLANs when they use dynamic channel bonding to access the channel. The analytical framework models the behavior of DCB networks using CTMCs. We devised an algorithm that is able to systematically construct the Markov chain corresponding to any DCB network. We then used our analytical model to explain some key properties of DCB networks--e.g., their sensitivity to the backoff and transmission time distributions, the behavior of dominant states, and the high switching times between different dominant states--all of which cannot be observed in networks of WLANs operating on a single shared channel, or even those using static channel bonding. Finally, we evaluated the analytical model in realistic scenarios, showing that it is able to give accurate results when some of the assumptions used in the analysis are relaxed. 


\appendix[Calculation of Packet Transmission Duration] \label{Apx:TX_duration}
Using the parameters in Tables \ref{Tbl:Parameters80211ac} and \ref{Tbl:RateAdaptation}, the time required for transmission of a packet over $n$ basic channels, $\frac{1}{\mu_{n}}$, can be calculated as
\begin{small}
\begin{align}
	\frac{1}{\mu_{n}}=2T_{\text{PHY}}+\left(\left\lceil \frac{\text{SF} + K_{\text{A}} (\text{MD}+\text{MH}+L_d) + \text{TB}}{L_{\text{DBPS}}(n)}\right\rceil T_s \right)
	+ T_{\text{SIFS}}+\left(\left \lceil \frac{\text{SF} + L_{\text{BA}} + \text{TB} }{L_{\text{DBPS}}(1)} \right \rceil T_s \right)+T_{\text{DIFS}}+T_{\text{slot}}
\end{align}\label{Eq:Ts}
\end{small}
where $T_{\text{PHY}}=40 \mu$s is the duration of the PHY-layer preamble and headers, $T_s=4~\mu$s is the duration of an OFDM (Orthogonal Frequency Division Multiplexing) symbol. SF is the \textit{service field} ($16$ bits), $\text{MD}$ is the \textit{MPDU Delimiter} ($32$ bits) MH is the \textit{MAC header} ($288$ bits), TB is the number of \textit{tail bits} ($6$ bits),  $L_{\text{BA}}$ is the \textit{Block-ACK} length ($256$ bits), $L_{\text{DBPS}}(n)=K_\text{m}R \xi(n)$ is the number of bits in each OFDM symbol, with $K_\text{m}$, $R$, and $\xi(n)$ are given in Table \ref{Tbl:RateAdaptation}. $K_{\text{A}}=64$ is the number of packets that are aggregated in each transmission. Finally, $T_{\text{DIFS}}=34~\mu$s and $T_{\text{SIFS}}=16~\mu$s are the DIFS and SIFS duration, respectively.


\bibliographystyle{unsrt}
\bibliography{TheBib}





\end{document}